\documentclass[9pt]{extarticle}
\usepackage[english]{babel}
    \usepackage{tikz-cd}
\usetikzlibrary{shapes.geometric}
\usetikzlibrary{positioning}
\usetikzlibrary{arrows}
\definecolor{mygreen}{RGB}{68,206,27}
\definecolor{myyellow}{RGB}{247,227,121}
\definecolor{myorange}{RGB}{242,161,52}
\definecolor{myred}{RGB}{229,31,31}
\usepackage{ragged2e}
\DeclareUnicodeCharacter{2217}{\star}

\usepackage[letterpaper,top=2cm,bottom=2cm,left=3cm,right=3cm,marginparwidth=1.75cm]{geometry}
\usepackage{bbm}
\usepackage{rotating}
\usepackage{mathtools}
\usepackage{graphicx}
\usepackage[utf8]{inputenc} 
\usepackage[T1]{fontenc}    
\usepackage{url}
\usepackage{booktabs}       
\usepackage{amsfonts}       
\usepackage{nicefrac}       
\usepackage{microtype}      
\usepackage{lipsum}
\usepackage{fancyhdr}       
\usepackage{graphicx}       
\graphicspath{{media/}}     
\usepackage{amsmath}
\usepackage{amsthm}
\usepackage{cases}
\usepackage{enumitem}
\usepackage{bm}

\usepackage{amssymb}
\usepackage{tikz}
\usetikzlibrary{arrows.meta}
\newtheorem{definition}{Definition}
 \usepackage{comment}
\usepackage{subfiles}

\usepackage{bbm}

\numberwithin{equation}{section}
\newtheorem{theorem}{Theorem}[section]

\newtheorem{lemma}[theorem]{Lemma}
\newtheorem{proposition}[theorem]{Proposition}

\newtheorem{corollary}[theorem]{Corollary}
\DeclareMathOperator*{\argmin}{arg\,min}
\newcommand{\gen}[1][]{s_{#1}}

\newcommand{\dem}[1][]{t_{#1}}
\newcommand{\flow}[1][]{f_{#1}}
\newcommand{\flowTot}{\bm{f}}
\newcommand{\flowLim}[1]{\bar{f}_{#1}}
\newcommand{\flowLimVec}{\bm{\bar{f}}}
\newcommand{\PR}[1]{\mathbb{P}\left(#1\right)}
\newcommand{\R}[0]{\mathbb{R}}
\newcommand{\N}[0]{\mathbb{N}}
\newcommand{\E}[0]{\mathbb{E}}
\newcommand{\EE}[0]{\mathcal{E}}
\newcommand{\V}[0]{\mathcal{V}}
\newcommand{\G}[0]{\mathcal{G}}
\newcommand{\K}[0]{\mathcal{K}}

\newcommand{\nv}[0]{n_{_\V}}
\newcommand{\nee}[0]{n_{_\EE}}
\newcommand{\nk}[0]{n_{_\K}}
\newcommand{\demM}[0]{\bm{T}}
\newcommand{\genM}[0]{\bm{S}}
\newcommand{\flowM}[0]{\bm{F}}
\newcommand{\incM}[0]{\bm{B}}

\newcommand{\paretoI}[1]{X_{#1}}
\newcommand{\pareto}[0]{\bm{X}}
\newcommand{\flowCost}[0]{c_f}
\newcommand{\flowMOpt}[1]{\bm{F}^*(#1)}
\newcommand{\flowMStar}[0]{\bm{F}^*}
\newcommand{\ind}{\perp\!\!\!\!\perp}

\newcommand{\OO}[1]{\mathcal{O}\left(#1\right)}

\newcommand{\exc}[1]{\psi_{#1}}

\newcommand{\SPM}[0]{\bm{G}}
\newcommand{\SPTot}[0]{\bm{g}}
\usepackage{tikz}
\usetikzlibrary{shapes.geometric, arrows}

\newcommand{\End}[0]{^{(end)}}
\newcommand{\unitV}[1]{\bm{e}_{#1}}
\usepackage{hyperref}
\hypersetup{
    colorlinks=true,
    linkcolor=blue,
    filecolor=magenta,      
    citecolor=red 
    }
\newcommand{\mytextcolor}[2]{%
  \ifx\relax#1\relax
    #2
  \else
    \textcolor{#1}{#2}%
  \fi
}

\newcommand{\hideRed}{%
  \renewcommand{\mytextcolor}[2]{%
    \ifx red##1
      \relax
    \else
      \textcolor{##1}{##2}%
    \fi
  }
}

\newif\ifhideRed

\let\oldtextcolor\textcolor
\renewcommand{\textcolor}[2]{%
  \ifhideRed
    \ifstrequal{#1}{red}{}{\oldtextcolor{#1}{#2}}%
  \else
    \oldtextcolor{#1}{#2}%
  \fi
}

\usepackage{etoolbox} 

\tikzstyle{startstop} = [rectangle, rounded corners, 
minimum width=2cm, 
minimum height=1cm,
text width = 2.5cm,
text centered, 
draw=black, 
fill=myred!20]

\tikzstyle{io} = [trapezium, 
trapezium stretches=false, 
trapezium left angle=70, 
trapezium right angle=110, 
minimum width=1cm, 
minimum height=1cm, text centered,
draw=black, fill=blue!30]

\tikzstyle{process} = [rectangle, 
minimum width=3cm, 
minimum height=2cm, 
text centered, 
text width=3cm, 
draw=black, 
fill=myorange!25]

\tikzstyle{decision} = [diamond, 
minimum width=3cm, 
minimum height=1cm, 
text centered, 
draw=black, 
fill=green!20]
\tikzstyle{arrow} = [thick,->,>=stealth]
\title{Emergence of Scale-Free Traffic Jams in Highway Networks:\\
\Large A Probabilistic Approach}

\author{
  \textbf{Agnieszka Janicka$^1$, Fiona Sloothaak$^{1}$, Maria Vlasiou$^{2}$, and Bert Zwart$^{1,3}$} \\\\
  $^1$\textit{Eindhoven Univeristy of Technology, 5612 AZ Eindhoven, the Netherlands}\\
  $^2$\textit{University of Twente, 7522 NB Enschede, the Netherlands}\\
  $^3$\textit{Centrum Wiskunde and Informatica, 1098 XG Amsterdam, the Netherlands}\\\\
}

\usepackage{subcaption}

\begin{document}
\maketitle

\begin{abstract}
\noindent Traffic congestion continues to escalate with urbanization and socioeconomic development, necessitating advanced modeling to understand and mitigate its impacts. In large-scale networks, traffic congestion can be studied using cascade models, where congestion not only impacts isolated segments, but also propagates through the network in a domino-like fashion. One metric for understanding these impacts is congestion cost, which is typically defined as the additional travel time caused by traffic jams. Recent data suggests that congestion cost exhibits a universal scale-free-tailed behavior. However, the mechanism driving this phenomenon is not yet well understood. To address this gap, we propose a stochastic cascade model of traffic congestion. We show that traffic congestion cost is driven by the scale-free distribution of traffic intensities. This arises from the catastrophe principle, implying that severe congestion is likely caused by disproportionately large traffic originating from a single location. We also show that the scale-free nature of congestion cost is robust to various congestion propagation rules, explaining the universal scaling observed in empirical data. These findings provide a new perspective in understanding the fundamental drivers of traffic congestion and offer a unifying framework for studying congestion phenomena across diverse traffic networks.
\end{abstract}

\flushbottom
\maketitle
\thispagestyle{empty}

\section{Introduction}
Traffic congestion is a persistent problem due to continued socioeconomic development and urbanization. Congestion is studied at multiple levels, from local analysis of intersections \cite{LIU2009412} and road segments \cite{roadSegmentCongestion}, through city-wide urban congestion \cite{Long_Gao_Ren_Lian_2008,
Ranjan2020}, to traffic jams in highway networks \cite{
salter1996highway, Avila_Mezić_2020}. These varying perspectives have led to many modeling approaches for congestion. Microscopic models focus on the dynamics and decision-making of individual drivers~\cite{mardiati2014review}, macroscopic models view the traffic flow from the perspective of fluid dynamics~\cite{lee2001macroscopic}, and mesoscopic models balance the micro and macro perspectives, aiming to retain the effects of individual behavior at a reduced modeling complexity \cite{Nagatani_2002}. 

Recently, network theory concepts, such as cascade models, have also been applied to study the global effects of congestion \cite{1,5,6,9}. Cascades (failures) appear in various complex systems, such as energy or water distribution systems \cite{Guo2017}, financial markets \cite{Huang_Vodenska_Havlin_Stanley_2013}, transportation networks \cite{Daqing_2014}, and more \cite{Watts_2002, Crucitti2004, Mahdi2017}. In addition, over the last 25 years, scale-free phenomena have been discovered in many real-life networks such as social networks, the internet, power networks, or airport networks \cite{barabasi1999emergence, int03, chen2001analysis, Guimer2004}. Recently, scale-free behavior has also been observed in congestion data from traffic networks \cite{ Chen_Lin_Yan_Liu_Liu_Li_2024, Zhang_Zeng_Li_Huang_Stanley_Havlin_2019, Serok_Havlin_Blumenfeld_Lieberthal_2022}. However, few insights exist on why this phenomenon emerges in traffic. In some studies, the scale-free phenomenon is attributed to self-organized criticality, although little supporting analysis is provided \cite{LAVAL2023104056, Nagatani_1995}. This paper offers the first and simple explanation on the emergence of scale-free traffic jams. 
A novel element of our approach lies in the application of a cascade model and probabilistic techniques to derive mathematically rigorous findings on the likelihood of large congestion.

 In this paper, we discover scale-free behavior in Dutch traffic jam data, as shown in Figure~\ref{fig:dataAnalysis}. Specifically, we use open access data on the length of traffic jams on highways from the National Road Traffic Data portal and we obtain that the empirical distribution of traffic jam lengths follows a scale-free distribution with scaling parameter $\alpha = 6.72$. This finding is alarming, as it indicates that large-scale traffic jams, while still rare, occur more often than conventional statistical laws would incur. Our study suggests that the scale-free nature of traffic jams could be inherited from the scale-free distribution of traffic intensities in the network. 
 \begin{figure}[t!] 
\begin{subfigure}{0.45\textwidth}
\includegraphics[width=\linewidth]{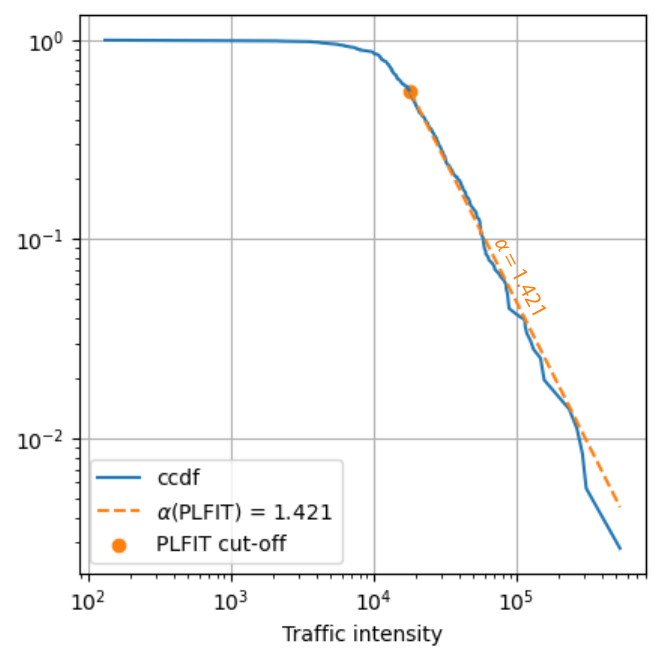}
\caption{CCDF plot of daily traffic intensity in Dutch municipalities measured in the number of car trips.} \label{fig:trafficVolumes}
\end{subfigure}\hspace*{\fill}
\begin{subfigure}{0.45\textwidth}
\includegraphics[width=\linewidth]{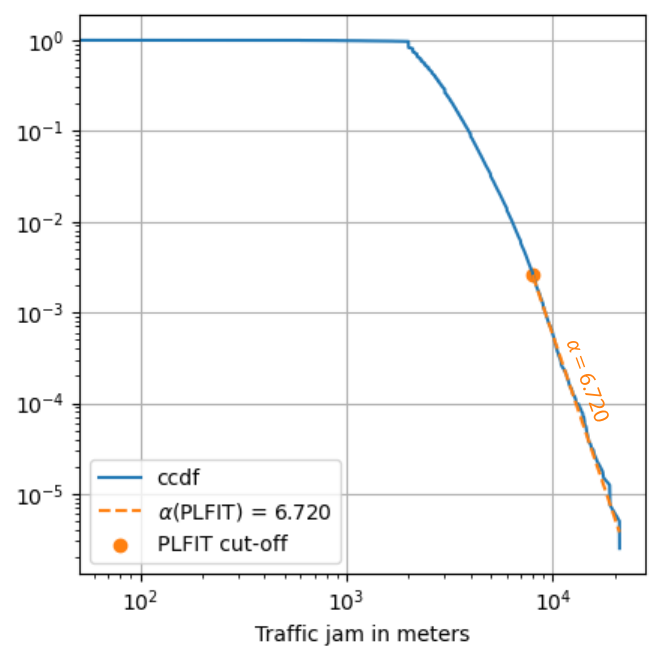}
\caption{CCDF plot of average length of traffic jams on Dutch highways measured in meters.} \label{fig:JamsDist}
\end{subfigure}
\caption{The traffic intensity and the average length of traffic jams on Dutch highways exhibit scale-free behavior. The full description of the data analysis is provided in Appendix~\ref{DataAnalysis}.}
\label{fig:dataAnalysis}
\end{figure}Notably, the distribution of city sizes, measured by the number of inhabitants, is often scale-free \cite{ROSEN1980165} and traffic intensity from a given location can be directly related to its population size. This correlation suggests that traffic intensity may inherit the scale-free property of city sizes. This is also consistent with West's scaling theory of cities \cite{Bettencourt_West_2010}, which postulates that the scale of urban processes, such as traffic intensity, is determined by the city's population size. In addition, our analysis of traffic intensities in the Netherlands reveals that they indeed have a scale-free tail (see Figure~\ref{fig:dataAnalysis}), which supports our novel hypothesis. Without a proper understanding of the formation of large traffic jams, it is difficult to anticipate their occurrences and devise strategies to prevent severe disruptions. Hence, understanding the underlying causes of large congestions becomes imperative. However, to the best of our knowledge, no study exists that explains the emergence of scale-free traffic jams in a systematic, methodological manner, which we address in this paper.

Cascade models are fundamental in studying congestion in traffic networks, particularly to facilitate the understanding and prediction of systemic failures \cite{Chen_Lin_Yan_Liu_Liu_Li_2024}, assess robustness \cite{1,9, 2}, measure vulnerability of network infrastructures \cite{7,8}, and develop congestion mitigation strategies \cite{5, 11,4}. These studies generally model traffic dynamics by means of travel intensities between origin-destination pairs and treat congestion as a complete failure of road segments \cite{1,6,2,7,8,11}. In contrast, our approach recognizes that congestion does not necessarily render a component unusable but may only lower its capacity, allowing for different levels of congestion. This approach, also followed in \cite{9}, reflects real-world behavior more accurately. Specifically, we assume that, in an initially stable network, a congestion cascade can occur due to an unexpected decrease in traffic capacity along a route -- e.g., in the form of a closed lane due to an accident. This may cause drivers to seek alternative routes, potentially causing additional congestion on those routes. This cycle of rerouting and additional congestion keeps on repeating, causing a cascade of road congestion, which results in profound implications on the system's behavior, frequently culminating in the disruption of its normal functioning \cite{DESHKAR2016}.

Flow distribution mechanisms model origin-destination traffic by describing how drivers traverse the network. Popular flow distribution mechanisms include the betweenness centrality-based mechanism \cite{6, 8, 3,Wang_Xu_Wu_2015}, the System Optimum (SO) \cite{9}, and the Wardrop User Equilibrium  (UE) \cite{1,9,2,7}. The UE approach ensures no driver can shorten their travel time by changing routes alone \cite{Kelly_Yudovina_2014}. This accurately reflects drivers behavior in highway networks, which is why we adopt this mechanism in our model. 

The impact of traffic jams is commonly defined as the additional travel time resulting from congestion \cite{7,4}. The travel time can be modeled using some flow cost function associated with the Wardrop UE flow problem, representing the cumulative travel time on all roads. This cost increases rapidly when the traffic flow exceeds the road capacity, effectively emulating the effects of congestion on travel times. Here, we define the congestion cost as the difference between the flow cost before and after the congestion cascade.

 This paper aims to explain the origin of the scale-free nature of traffic jams, which we also observe in Dutch traffic data. To this end, we propose a stochastic cascade model for highway congestion, detailed in Section~\ref{modDesc}. The model considers a graph $\G = (\V,\EE)$ representing the highway network, with vertices (cities) and edges (highways), as input. We assign a weight to each vertex, representing the amount of traffic leaving it, which determines the initial Wardrop UE flow. This flow is destabilized by a random disruption event that results in a reduced capacity of an edge, causing a cascade of congestion events. The resulting congestion cost ($\Delta c_f\End$) represents the cumulative added travel time caused by the congestion cascade.

Our main result, given in Section \ref{results}, shows that if the vertex weight distribution is scale free, then the congestion cost is also scale free with the same scaling parameter. More precisely, for a vector of vertex weights $\pareto = (\paretoI{v})_{v\in \V}$ where all $\paretoI{v}$'s are mutually independent and satisfy
\begin{equation}\PR{X_v>x} \approx K\cdot x^{-\alpha}, \quad~\forall ~v\in \V, \label{eq:scale-freeIntro}\end{equation} for $\alpha$, $K>0$ and $x$ large enough, we show that \[\PR{\Delta c_f\End>x}\approx C\End\cdot x^{-\alpha},\] for some constant $C\End>0$. This shows a direct relationship between the traffic intensity and congestion. From our analysis, it is also evident that the cascade mechanism has only a secondary effect on the tail distribution of the congestion cost; it can only affect the constant $C\End$, but it does not change the underlying power-law parameter. This finding emphasizes the prevalence of power-law scaling in traffic networks. 

In Section \ref{roadmap}, we discuss the key ingredients used to prove our results: first, the scale-invariance property, showing that scaling the initial vertex weights by a constant results in scaling of the congestion cost by the same constant, then
the continuity of the congestion cost function, implying that a small perturbation of the vertex weights results in a small change in congestion cost,
and finally the catastrophe principle for the congestion cost function, showing that with high probability, large traffic jams occur when a single vertex has much larger weight than all other vertices \cite{nair_wierman_zwart_2022}.  

Our model employs a parametric cascading congestion mechanism. This allows us to study different congestion propagation rules and to vary the impact of congestion on the capacity of roads. Importantly, our results depend on the choice of these parameters only through the constant $C^{(end)}$, suggesting that the scale of the congestion phenomenon is independent of the specifics of the congestion mechanism. This finding offers a plausible explanation for the universal scaling observed in the data across different locations \cite{Zhang_Zeng_Li_Huang_Stanley_Havlin_2019}.
\section{Results}\label{results}
We begin this section with a high-level description of our model. For a complete description, we refer the reader to Section~\ref{modDesc}. Then, we state and discuss the main result, Theorem~\ref{mainTheorem}, and show the catastrophe principle for our model in Proposition~\ref{sbj}. These results, along with supplementary lemmas, are proven in Appendix~\ref{analysis}. We conclude this section with simulation results for a small network, which helps visualize the cascade dynamics and validates our theoretical results.

\begin{figure}[!b]
    \centering
        \resizebox{0.75\textwidth}{!}{%
\begin{tikzpicture}[node distance=1.7cm]

\node (start) [startstop] {Initial edge disruption\\
$(r = 1)$};
\node (pro1) [process, right of=start, xshift = 2.5cm] {\textbf{Step 1}\\Determine the set of disrupted edges}
;
\node (dec1) [decision, right of=pro1, xshift=2.5cm] {Is the set empty?}
;
\node (pro3) [process, below of=pro1, yshift=-2cm] {\textbf{Step 3} \\Compute the new Wardrop UE flows};

\node (pro2) [process, below of=dec1, yshift=-2cm] {\textbf{Step 2} \\Adjust edge capacities and disruption states};

\node (stop) [startstop, right of=dec1, xshift = 2.5cm] {Cascade terminates};

\draw [arrow] (start) -- node[anchor = south] {$r = 2$}(pro1);
\draw [arrow] (pro1) -- (dec1);
\draw [arrow] (pro2) -- (pro3);
\draw [arrow] (pro3) -- node[anchor=east] {$r\mapsto r+1$}(pro1);
\draw [arrow] (dec1) -- node[anchor=south] {yes} (stop);
\draw [arrow] (dec1) -- node[anchor=east] {no} (pro2);

\end{tikzpicture}}
    \caption{Flow diagram of disruption cascade. One loop in the process corresponds to a single cascade stage, denoted by $r$.}
    \label{flowDiagram}
\end{figure}
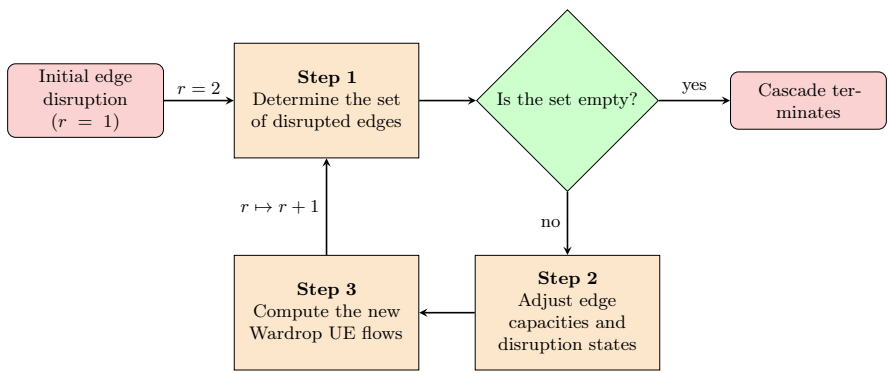We proceed with a brief description of the model. We consider a graph $\G = (\V,\EE)$, where each vertex $v\in \V$ has an associated scale-free weight $X_v$, as given in \eqref{eq:scale-freeIntro}, which defines the intensity of traffic leaving $v$. A cascade operates in stages indexed by $r$, according to the diagram in Figure~\ref{flowDiagram}. Each stage has associated edge flow capacities $\flowLimVec^{(r)}$  and Wardrop UE flow $\flowM^{(r)}$, which yield flow costs $c_f(\flowM^{(r)}, \flowLimVec^{(r)})$. Stage $r = 0$ corresponds to a non-congested network flow. At stage $r = 1$, the initial disruption occurs at a randomly chosen edge and the flow capacity of this edge is reduced by a factor $\phi_{init} = u$ with probability $p_{init}(u)$ for $u\in [l_m, l_M]\subset (0,1)$. This yields the updated edge flow capacity $\flowLimVec^{(1)}$, and by employing the Wardrop UE, we obtain the updated network flow $\flowM^{(1)}$. At stage $r>1$, any edge $e\in \EE$ whose flow exceeds the current flow capacity of the edge may experience a disruption with probability $p_e^{(r)}$, which is an increasing, continuous function of the edge flow. This results in an additional reduction of the capacity of edge $e$ by a factor $\phi_e^{(r)}$. This yields the updated edge flow capacity $\flowLimVec^{(r)}$ and the Wardrop UE flow $\flowM^{(r)}$. The cascade ends when no edge disruptions occurs at a certain stage. The congestion cost at stage $r$ is given by $\Delta c_f^{(r)} := c_f(\flowM^{(r)}, \flowLimVec^{(r)}) - c_f(\flowM^{(0)}, \flowLimVec^{(0)})$ and $\Delta c_f\End$ denotes the congestion cost at the final stage of the cascade. For a precise definition of the function $c_f$ see Equation~\eqref{costEquation}.

Our main result quantifies the probability that congestion cost exceeds a high level. In particular, we show that at every stage of the cascade, the probability that the congestion cost exceeds $y$ behaves approximately as $C\cdot y^{-\alpha}$ for some cascade stage-specific constant $C$ and a large value of $y$. Moreover, the scale parameter $\alpha$ is the same as the scale parameter governing the distribution of vertex weights. This implies that the scale-free behavior of the congestion cost arises as a consequence of the vertex weight distribution. The following theorem makes the result explicit.

\begin{theorem} \label{mainTheorem}
Let $\G = (\V, \EE)$ be a graph, and assume that the distribution of vertex weights is Pareto-tailed with parameter $\alpha$ and normalization constant $K$, as given in Equation~\eqref{eq:scale-freeIntro}. Then, 
\begin{enumerate}
    \item At every stage of the cascade, the congestion cost has a scale-free tail with parameter $\alpha$. Specifically, for every $r\in \N$,
    \begin{equation}
        \label{eq:mainThm}
    \PR{\Delta c_f^{(r)}>y} \sim C^{(r)} y^{-\alpha},\quad \text{as }y\rightarrow \infty, \quad ~i.e., ~ \lim_{y\rightarrow \infty }\PR{\Delta c_f^{(r)}>y}y^{\alpha} = C^{(r)}\end{equation}
    with 
    \begin{equation}C^{(r)} := K\sum_{v\in \V} \E\left[\left(\Delta \flowCost^{(r)}\right)^\alpha\,\Big|\,X_v = 1, X_w = 0,\quad \forall \,w\in \V\setminus\{v\}\right]>0.\label{eq:Cr}\end{equation}

    \item The congestion cost caused by the entire congestion cascade also exhibits scale-free tail behavior with parameter $\alpha$, i.e.,
\begin{equation}
\PR{\Delta c_f^{(end)}>y} \sim C^{(end)} y^{-\alpha},\quad \text{as }y\rightarrow \infty,   \label{eq:DeltaEnd}
\end{equation}
    with $C^{(end)} = \lim_{r\rightarrow \infty} C^{(r)}>0$.
\end{enumerate}
\end{theorem}
Notably, from the expression of the constant $C^{(r)}$ in \eqref{eq:Cr}, it is evident that, asymptotically, the probability of large congestion can be decomposed into independent contributions from each vertex of the graph. One can view the contribution of vertex $v$ as the $\alpha$-th moment of the cascade cost at stage $r$, given that all traffic originates from the vertex $v$. This is a consequence of the catastrophe principle for scale-free distributions \cite{nair_wierman_zwart_2022}, implying that, typically, a single rare event (here, one large vertex) is responsible for extreme behavior of the entire process (i.e., the cascade cost). In the context of our model, this principle can be formalized as the proposition below. Note that for two functions $g(x)$ and $h(x)$ we say that $g(x) = \OO{h(x)}$ as $x\rightarrow \infty$ if and only if $\exists~N, x^*>0$ such that $|g(x)|\leq N|h(x)|$ for all $x\geq x^*$.
\begin{proposition} \label{sbj}
    Consider a vector of vertex weights $(X_1,\dots, X_{\nv})$. Let $X_{(\nv)} := \max\{X_1,\dots, X_{\nv}\}$. Then, for all $\varepsilon>0$ and $r\in \N$,
    \[\PR{ \Delta c_f^{(r)}>y ,\ \sum_{i = 1}^n X_i - X_{(\nv)} \geq \varepsilon X_{(\nv)}} = \OO{y^{-2\alpha}}\quad \text{as }y\rightarrow \infty.
    \]
\end{proposition}
This result, together with the main theorem, suggests that large traffic jams are more likely to occur in networks with a single large city.

It is important to note that in our model we parametrize the functions driving the cascade behavior, which enables us to model various cascade mechanisms. Specifically, for each edge $e\in \EE$ and cascade stage $r\in \N$, we allow for a different continuous probability law $p_e^{(r)}$ and reduction factor $\phi_e^{(r)}$, both of which together determine the likelihood and severity of the edge disruption. In Theorem~\ref{mainTheorem}, only constants $C^{(r)}$ and $C^{(end)}$ depend on the choice of these parameters, implying that the scale-free behavior of the cascade cost is robust to the cascade mechanism. 
\begin{figure}[!t]
  \centering
  \begin{minipage}[b]{0.57\textwidth}
  \centering
\includegraphics[width=0.9\textwidth]{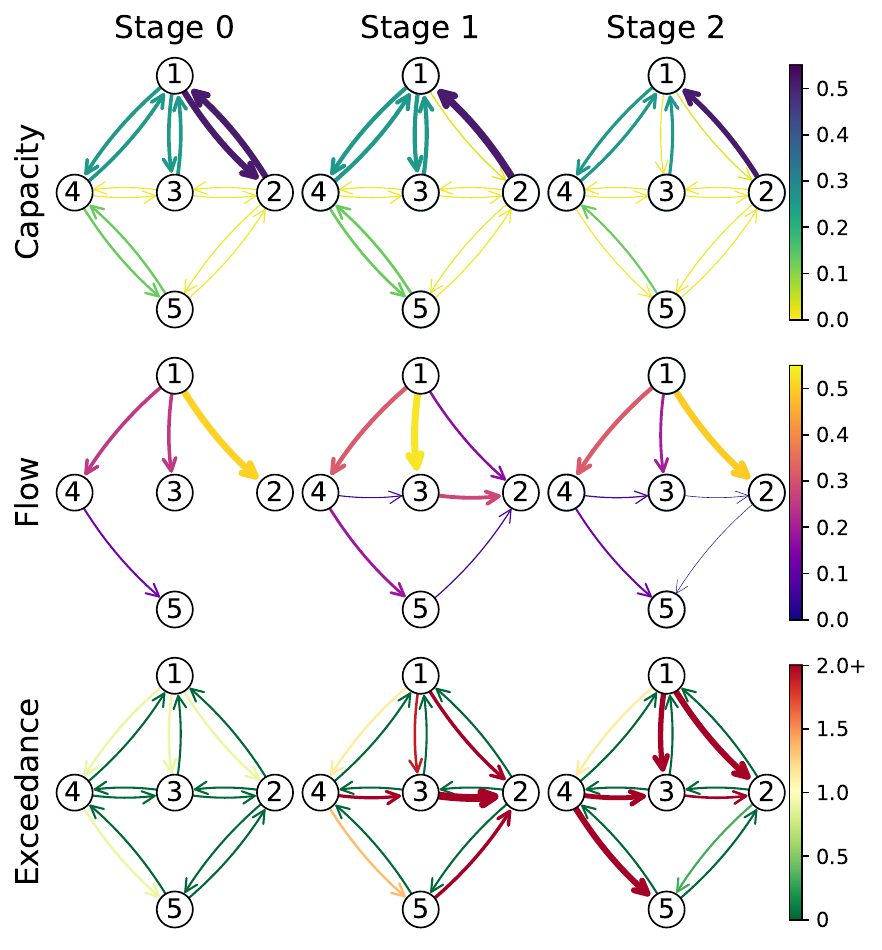}
         \caption{A cascade example initiated by the congestion on edge $(1,2)$. Here, edge exceedance is the fraction of the flow over flow capacity on the edge.}
         \label{cascade}
  \end{minipage}
  \hfill
          \nextfloat
  \begin{minipage}[b]{0.42\textwidth}\centering
      \begin{subfigure}[b]{\textwidth}
        \centering
\includegraphics[width=\textwidth]{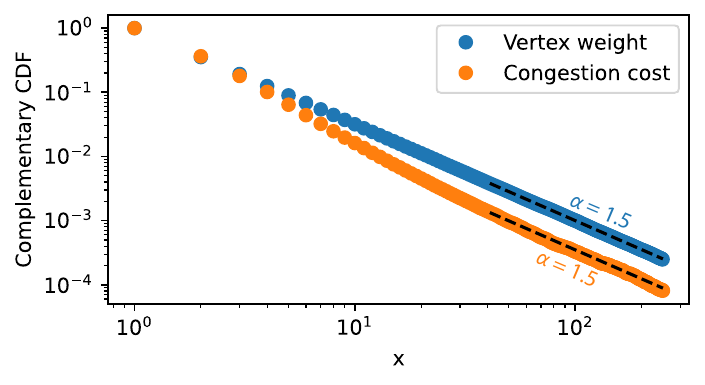}
    \caption{The CCDF of empirical distributions for the vertex weight $\bm{X}$ and congestion cost $\Delta \flowCost\End$.}
\label{loglogdist}
     \end{subfigure}
     \hfill\\
     \begin{subfigure}[b]{\textwidth}
         \centering
\includegraphics[width=\textwidth]{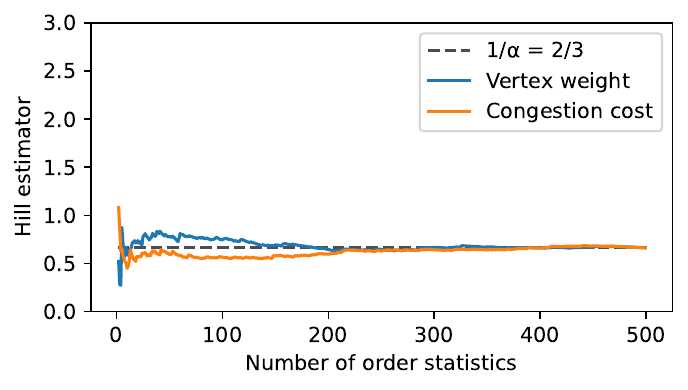}
         \caption{Hill estimators for the vertex weight $\bm{X}$ and congestion cost $\Delta\flowCost\End$.}
         \label{HillPlot}
\end{subfigure}
\caption{Congestion simulation in the example graph.}
        \label{exampleCong}
  \end{minipage}

\end{figure}

Having stated the main results, we proceed to demonstrate them on a small simulated network example, detailed in Appendix~\ref{example} and depicted in Figure~\ref{cascade}. This figure shows a single realization of a cascade on the graph. As stated earlier, an insight from our theoretical results is that asymptotically, large congestion occurs when one vertex has a large weight and the remaining weights approach 0. Hence, in the example, we consider the asymptotic normalized scenario with vertex weight vector $\bm{X} = (1,0,\dots,0)$. The intensity of the traffic from the vertex $1$ to $i$ is given by $q_{1,i}$ such that $(q_{1,1},q_{2,1}, q_{3,1},q_{4,1},q_{5,1}) = (0,1/2,1/4,1/8,1/8)$, and it is 0 from all other vertices. Moreover, for brevity, we assume that each edge can be disrupted at most once. In the first column of Figure~\ref{cascade}, we depict the cascade stage 0, where the network is stable and without congestion. The cascade is initiated by a disruption on edge (1,2), which causes a fraction of travelers who previously took route $1\rightarrow 2$ to choose one of the routes $1\rightarrow 3\rightarrow 2$ and $1\rightarrow 4\rightarrow 5\rightarrow 2$ instead. This leads to additional congestion on edges (1,3), (3,2), (4,3), (4,5), and (5,2), which become disrupted in stage 2. This results in the redistribution of the flow. At this point, all edges that have not yet experienced disruption have flow that does not exceed the edge's capacity, which means that the cascade terminates.

Next, we demonstrate the findings of Theorem \ref{mainTheorem}. Figures~\ref{loglogdist} and~\ref{HillPlot} present the results of $N = 10^6$ cascade simulations on the graph where vertex weights follow a Pareto distribution with parameter $\alpha = 3/2$. From the log-log plot, it is evident that the distribution of the congestion cost $\Delta \flowCost\End$ has a scale-free tail with parameter $\alpha$, inherited from the vertex weight distribution, which validates our theoretical findings. The estimated constant $C\End$ is equal to 0.35. The Hill plot also confirms that both distributions have a scale-free tail because we observe that both plots flatten around the value $2/3 = \alpha^{-1}$ for the choice of order statistics larger than 200.

Lastly, to visualize the catastrophe principle for our model, namely Proposition \ref{sbj}, we compare the cascade cost probabilities for two scenarios:
\begin{itemize}
    \item \textbf{Scenario 1}: $\sum_{i = 1 }^{\nv} X_i>2\max\{X_1,\dots,X_{\nv}\}$. This means that the total vertex weight is at least twice the maximum vertex weight.
    \item \textbf{Scenario 2}: $\sum_{i = 1 }^{\nv} X_i\leq 2\max\{X_1,\dots,X_{\nv}\}$. This means that the contribution of all other vertices is relatively small compared to the maximum.
\end{itemize}
\begin{figure}[b]
     \centering
     \begin{subfigure}[b]{0.49\textwidth}
         \centering
         \includegraphics[width=0.9\textwidth]{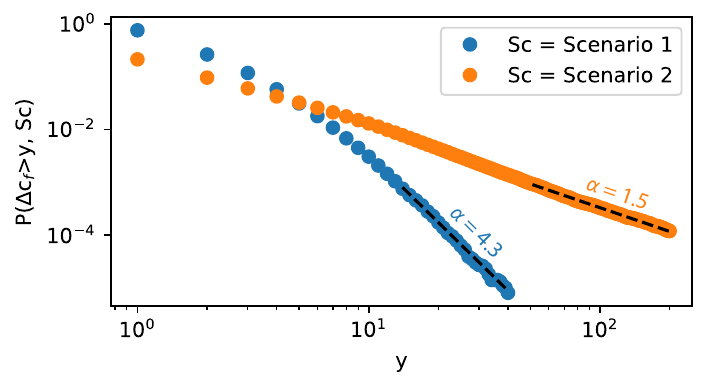}
         \caption{Probability that the cascade cost $\Delta \flowCost \End$ exceeds~$y$.  }
         \label{loglogdist1}
     \end{subfigure}
     \hfill
     \begin{subfigure}[b]{0.49\textwidth}
         \centering
         \includegraphics[width=0.9\textwidth]{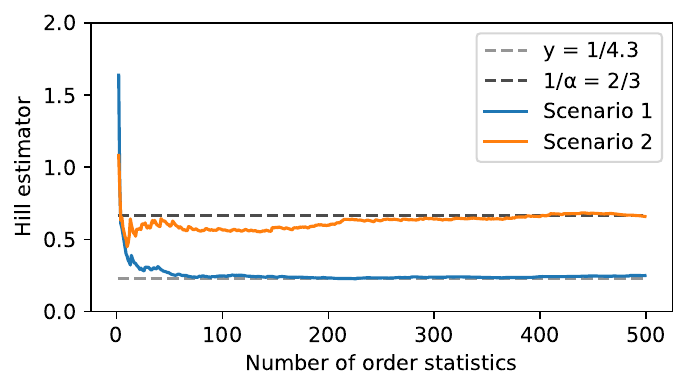}
         \caption{Hill estimators for the congestion cost $\Delta \flowCost \End$.}
         \label{HillPlot1}
     \end{subfigure}
        \caption{Plots of congestion simulation in the example graph for two vertex weight scenarios.}
        \label{Fsbj}
\end{figure}
Proposition~\ref{sbj} tells us that the probability that the cascade is larger than $y$ under Scenario 1 is of order $\OO{y^{-2\alpha}}$. This, together with the result of Theorem~\ref{mainTheorem}, implies that the $\alpha$-tail behavior occurs due to vertex weights that fall into Scenario 2. This is also evident from the simulation results. In Figure~\ref{loglogdist1}, we observe that under Scenario 1, the tail of the congestion cost behaves approximately as $\Theta(y^{-4.3}) = \OO{y^{-2\alpha}}$, whereas for Scenario 2, we again observe the scale-free tail behavior with parameter $\alpha$. The same conclusion can be drawn from Figure~\ref{HillPlot1}, where we observe that the curves flatten around values 100 and 400 for Scenarios 1 and 2, respectively, confirming the scale-free behavior observed in Figure~\ref{loglogdist1}. 

\section{Discussion}\label{discussions}
This study provides a novel explanation for the emergence of scale-free traffic jams in highway networks. We introduce a cascade model for highway congestion and prove that the scale-free nature of traffic jams arises from scale-free traffic intensity. Our analysis builds on the natural assumption—supported by data in Figure~\ref{fig:1}—that traffic intensities follow a scale-free distribution, establishing the first causal link between traffic intensity and congestion. Using probabilistic techniques, in Theorem~\ref{mainTheorem} we derive rigorous asymptotic results on the probability of large traffic jams within a cascading framework that not only aligns with, but also builds upon existing cascade models in traffic congestion. Our findings not only deepen the understanding of traffic jam formation, but also offer the first mathematically rigorous explanation of scale-free congestion phenomena.

Additionally, our results improve the understanding of  scale-free congestion phenomena beyond explaining their emergence. Specifically, in Proposition \ref{sbj} we show that among all traffic networks with scale-free city-size distributions, extreme traffic jams are more likely to occur because most traffic originates from one location, rather than other events. This key insight can help identify highway networks that are more prone to severe congestion. 
     
Furthermore, our findings imply that the scale-free nature of traffic jams is robust to the choice of network configuration or the specific congestion propagation mechanism. Specifically, as shown in Theorem \ref{mainTheorem}, they only influence the prefactor $C^{(r)}$, given in \eqref{eq:Cr}, and not the scale-free parameter $\alpha$. The robustness of the scale-free behavior of traffic jams to network configurations and propagation mechanisms may explain why similar scaling parameters were recovered from traffic jam data across different locations \cite{ Chen_Lin_Yan_Liu_Liu_Li_2024,Zhang_Zeng_Li_Huang_Stanley_Havlin_2019}.

Our work also has important implications for congestion prevention. A key insight from the robustness of the scale-free behavior is that mitigating the scale-free cost of traffic jams may be challenging without altering the fundamental distribution of traffic intensities, which is inherently difficult to control. However, strategic network augmentation or targeted traffic management strategies can help reduce the prefactor $C^{(r)}$. In large networks, this prefactor is expected to be significant, meaning that even a modest reduction could have a substantial impact on congestion severity. Still, the exact effects of prevention strategies remain an open question for future exploration.

Notably, our framework is inspired by a similar line of work on power transmission systems \cite{Nesti_2020, janicka2024scalefree}. There, the authors propose a power flow model and show that, similarly to traffic congestion, scale-free blackout sizes may be a consequence of the scale-free distribution of cities. Both models are examples of societal networks connecting cities and operate according to the same general principles, but they differ significantly regarding detailed flow and cascade dynamics. Despite these differences, the underlying root cause of ``severe'' cascades is the same in both models. Given this, and the increasing reliance of transportation systems on electrical power, as exemplified by the electrification of vehicles \cite{Lv_21}, it may be beneficial to explore potential hidden dependencies that arise from their connection to cities. This exploration could be effectively conducted using a multilayer cascade model~\cite{BOCCALETTI20141,Buldyrev_Parshani_Paul_Stanley_Havlin_2010, Vespignani_2010, DeDomenico_2015, Bianconi_2014, Su_Li_Peng_Kurths_Xiao_Yang_2014}. Moreover, cities drive the demand in various other flow networks, including gas and water distribution systems, hence incorporating the scale-free characteristics of city sizes into models of these systems could reveal critical new insights.

The key element of our model is the flow cost function, which determines the Wardrop UE flow at every stage of the congestion cascade. Future modifications could explore the impact of different cost functions. Common cost functions described in the literature include linear or convex polynomial functions of the flow. We anticipate that the application of such cost functions would maintain the scale-free nature of traffic jams, albeit potentially with a modified scaling parameter. It is naturally possible to create cost functions that would counteract the emergence of scale-free traffic jams. For example, cost functions of at most logarithmic order may lead to this effect. However, such functions typically model the principle of \textit{economies of scale}, where each increase in demand becomes significantly cheaper relative to the previous one. This assumption is not natural in the context of highway traffic and is not supported by data.

Furthermore, our model operates on a timescale that captures the emergence of traffic jams; however, a larger timescale is necessary to model congestion resolution. To model this larger timescale, it would be essential to incorporate dynamic traffic intensities that reflect user responses to road conditions and daily fluctuations in traffic volume, as well as road capacity recovery facilitated by, for example, incident management teams. Such an extension would enable a detailed study of the characteristics of congestion duration and recovery.

  Lastly, this work presents a novel explanation for the emergence of congestion; however, to further substantiate our findings, validation with empirical data is essential. Since our model aims to identify general traffic trends rather than detailed behavior, we expect our theoretical results and empirical data to exhibit similar trends, though perfect agreement is not anticipated. Our preliminary analysis of the Dutch highway network indeed reveals scale-free behavior in both traffic intensity and traffic jam length, as expected based on our theorem. However, the available data has certain limitations. First, it employs a different congestion metric than our model; the data is given in traffic jam lengths and our results in cumulative added travel time. Given the discussion above on the influence of the congestion cost function, this discrepancy may explain why we recover different scaling parameters for traffic intensity and congestion. 
  Second, the Dutch highway network is relatively small, i.e., of order $10^2$ vertices. Research suggests that accurately recovering scale-free behavior from data typically requires a sufficiently large network, generally on the order of $10^4$ vertices or more \cite{Nesti_2020}. In smaller networks, data may indicate a lighter tail than the true distribution or fail to exhibit a scale-free tail altogether. Therefore, further analysis of data from a larger traffic network is necessary to numerically validate our findings.

\section{Methods}
\label{methods}
In this section, we first introduce the general notation. We then detail our model in Section~\ref{modDesc}, where we first outline the flow network components, after which we explain the dynamics of how congestion propagates through the network. Then, in Section~\ref{roadmap}, we discuss the main ideas used in the proof of our main result, Theorem~\ref{mainTheorem}.

Throughout this paper, we adopt the following notation. We denote deterministic matrices and vectors with bold capital and small letters, respectively i.e., $\bm{A}\in \R^{n\times m}$ and $\bm{a}\in \R^n$. Their elements are denoted by $a_{i,j}$ and $a_{i}$, respectively, with $i\in \{1,\dots, n\}$ and  $j\in \{1,\dots, m\}$. The $n\times n$ identity matrix is denoted by $\bm{I}_n$. In addition, we use $\bm{e}$ to denote all-ones vector of appropriate size and letter $\bm{e}_j$ to denote the $j$-th unit vector of appropriate size with 1 at position $j$ and zero otherwise. The set of all non-negative real numbers is denoted by $\R_+$ and the set of all natural numbers, excluding 0, is denoted by $\N$. Lastly, for two functions $f(x)$, $g(x)\in \R$, we say that $f(x)\sim g(x)$ as $x\rightarrow \infty$ if and only if $\lim_{x\rightarrow \infty} f(x)/g(x) = 1$.

\subsection{Model} \label{modDesc}

\subsubsection{Flow network components} \label{basicModelComp}

\textit{Road network structure:} We represent a highway network through a directed, connected graph $\G = (\V, \EE)$, where $\V$ and $\EE$ represent the set of vertices and edges of $\G$, respectively, with the corresponding set sizes denoted by $\nv:=|\V|$ and $\nee := |\EE|$. Here, edges correspond to highways, and vertices correspond to highway crossings. Highways typically have lanes in both directions, hence we assume that the graph is symmetric, i.e., all edges appear twice, one in each direction. The graph is described through its incidence matrix $\incM\in \{-1,0,1\}^{\nv\times\nee}$, given by
\[b_{v,e} := \begin{cases} 1 & \text{if edge } e\in \EE \text{ enters vertex } v\in \V,\\
-1 & \text{if edge } e\in \EE \text{ exits vertex } v\in \V,\\
0 &\text{otherwise.}
\end{cases}\]
Each vertex $v\in \V$ has a stochastic weight $X_v$ with $\bm{X} = (X_v)_{v\in \V}$. We associate vertices with cities and vertex weights with city sizes. As city sizes are known to have a \textit{scale-free tail} \cite{ROSEN1980165, nair_wierman_zwart_2022,  Beckmann_1958, Berry1961}, we assume that
\begin{equation}\mathbb{P}(X_v>x) \sim Kx^{-\alpha} \quad \text{as }x\rightarrow \infty,\label{Kconstant}\end{equation}
 for some constants $K,\alpha>0$. Additionally, we assume that the vertex weights are mutually independent; i.e., $X_v\ind X_w$ for all $v\neq w \in \V$.
 
\vspace{0.25cm}
\textit{Origin and destination of commuters:} In highway networks, commuters travel from an origin vertex to a destination vertex. As common in the literature \cite{mCMF1,Barnhart2009, Whittle_2007}, we view the flow of traffic as a multi-commodity network flow. In this model, commuters originating from some vertex $k$ constitute a separate \textit{commodity} $k$, but our approach could be generalized to include other commodities, for example, vehicle type. Note that for all quantities $a$, for which the commodity is relevant, we use notation $a_{l,k}$ where $l$ refers to the location (either a vertex or an edge) and $k$ to the commodity. We denote the amount of traffic of commodity $k$ originating at vertex $v$ as $\gen[v,k]$ and the amount of traffic of commodity $k$ with destination $w$ as $\dem[w,k]$; this notation is inspired by the classical $s$-$t$ flow concept.  The inhabitants of vertex $v$ that remain in the vertex can be viewed as traffic from $v$ to $v$. Moreover, only commuters of commodity $k = v$ can originate from $v$. Hence, $\genM = (s_{v, k})_{v,k\in \V} = \text{diag}(\pareto)$, i.e., the diagonal matrix with $\pareto$ on the diagonal. 
Next, let the fraction of the commodity $k$ that travels to vertex $w$ be denoted by $q_{w,k}\geq 0$ such that $\sum_{w\in V} q_{w,k} = 1$. Using this, we obtain that the amount of traffic into $w$ of commodity $k$ is given by\begin{equation}\dem[w, k] = q_{w,k}\gen[k,k].\label{trafficIntensity}\end{equation}
In matrix notation, Equation~\eqref{trafficIntensity} can be written as
$\demM=  \bm{Q}\genM.
$ We call $\genM$, $\demM$, and $\bm{Q}$ \textit{origin}, \textit{destination}, and \textit{travel factor} matrices and the difference between the destination and origin matrix is the \textit{net travel} matrix $\bm{U}:= \demM - \genM$. 

\vspace{0.25cm}
\textit{Network flow:}
Commuters induce a flow on the network. Let $\flowM\in \R_+^{\nee\times \nv}$ be the \textit{flow} matrix where $\flow[e,k]\geq 0$ denotes the amount of commodity $k$ transported through edge $e$. The \textit{total flow} vector $\flowTot\in \R_+^{\nee}$ represents the cumulative flow on each edge and is given by $\flowTot:= \flowM \bm{e}$, were $\bm{e}$ is the all-ones vector of size $\nv$. The \textit{edge capacity }vector $\flowLimVec\in \R_+^{\nee}$ represents the capacity of each edge. If the total flow on an edge exceeds its capacity, it may lead to more disruptions, as will be explained in more detail in Section~\ref{emergencyPhase}.

We assume that the flow in the network is dictated by the Wardrop User Equilibrium (UE) flow. UE refers to a state of the network where no individual commuter can reduce their travel time by choosing a different route from their origin to their destination. In other words, the network is in equilibrium and all users commute on their optimal route. The flow matrix $\flowM^*$ corresponding to UE can be obtained by minimizing the flow cost function $c_f(\flowM, \flowLimVec)$, subject to the constraint that every commodity reaches its destination. In particular, the Wardrop UE flow is the optimal solution to \cite[p.\ 97]{Kelly_Yudovina_2014}:
    \begin{subequations} \begin{gather}\tag{F}\label{optFlowProblem} \flowM^*\left(\bm{U}, \flowLimVec
    \right):= \argmin_{\flowM\in \R^{n_{\EE}\times n_{\V}}_+}\flowCost(\flowM, \flowLimVec) \\
s.t.~~\incM \flowM = \bm{U}\tag{F1}\label{flowBalance}.\end{gather}\end{subequations}
where $c_f(\flowM, \flowLimVec)$ is the flow cost function. Constraint~\eqref{flowBalance} asserts that for all vertices $v,k\in\V$, the difference between the total \textit{flow into} $v$ and total \textit{flow out of} $v$ of commodity $k$ is equal to the net travel requirement at vertex $v$ of commodity $k$. As is commonly used in traffic literature \cite{1}, we assume that the flow cost function is given by
\begin{equation}c_f(\flowM, \flowLimVec) = \sum_{e\in \EE} \left(d_e\flow[e] + \frac{b_e}{\beta+1}\flowLim{e}\left(\flow[e]/\flowLim{e}\right)^{\beta+1}\right).\label{costEquation}\end{equation}
In the above cost function, $d_e>0$ is the free-flow travel parameter of edge $e$, representing the time required to travel through edge $e$, given no impediments. The parameter $b_e > 0$ quantifies the impact of congestion on edge $e$. This reflects how increased congestion might have a more significant effect on a narrow road compared to a wider one. The parameter $\beta>0$ models the global effect of congestion on travel time as an exponential function. As such, when the flow is below the capacity, the impact of congestion remains minimal; however, it increases drastically once the capacity is exceeded, accurately reflecting real-world traffic dynamics. Note that the exponent $\beta+1$ originates from an alternative representation of the Wardrop UE problem, discussed in \cite{Kelly_Yudovina_2014}.

Since $\beta>0$, the flow cost function is strictly convex, which we establish rigorously in Lemma~\ref{LemmaStrictConvex} of Appendix~\ref{appStrictConvex}. Moreover, Problem~\eqref{optFlowProblem} is feasible because $\G$ is connected and $\bm{e}^T\bm{U} = 0$. Hence, $\flowM^*$ is unique for any finite input as feasibility and strict convexity guarantee the existence and uniqueness of a solution of convex optimization problems \cite[ch.\ 4]{boyd2004convex}. Note that flow $\flowM^*(\bm{U}, \infty)$ is equivalent to solving Problem~\eqref{optFlowProblem} with cost function $\sum_{e\in \EE} d_e \flow[e]$, which can be viewed as the shortest path problem on $\G$ with edge weights $d_e$. As this particular flow occurs frequently in our analysis, we use the shorthand notation $\SPM:= \flowMStar(\bm{U}, \infty)$ and $\SPTot := \SPM\bm{e}$. Note that $\SPM$ may not be unique, in which case we assume that $\SPM$ is the optimal solution of~\eqref{optFlowProblem} that assigns equal flow to all shortest paths between vertices $v$ and $w$ for all $v,w\in \V$. This is discussed in more detail in Appendix~\ref{appendixSP} where we also show the equivalence of $\flowM^*(\bm{U}, \infty)$ with the shortest path flow in Lemma~\ref{SP}. 

\vspace{0.25cm}
\textit{Edge capacities:} We determine the edge capacity vector by considering the optimal flow without restrictions. In other words, we use the Wardrop UE flow with infinite capacity, i.e. the shortest path flow matrix $\SPM$. Then, we set the capacity on edge $e$ to be the maximum of i) minimal edge capacity $\flowLim{\min}$, ii) fraction $\tau$ of total shortest path flows on edges $e =(v,w)$ and $\tilde{e} = (w,v)$. In other words, 
\begin{equation}\flowLim{e} = \max\{\flowLim{\min},\tau g_{e}, \tau g_{\tilde{e}}\},\label{capacityPl}\end{equation} with $\tau\geq 1$ and $\flowLim{\min} = \varepsilon_{\min}\sum_{v\in \V} \gen[v,v]$, for a constant $\varepsilon_{\min}>0$. This choice captures the following natural properties:
\begin{enumerate}
    \item \textbf{Symmetry:} The capacities of edges $e$ and $\tilde{e}$ are equal, reflecting that highways typically have the same number of lanes in each direction.
    \item \textbf{Minimal capacity:} Minimal capacity threshold $\bar{f}_{\min}$ is chosen such that it scales linearly with the total weight of the vertices. As a consequence, edge capacities $\flowLim{e}$ also scale with the city sizes. 
    \item \textbf{Flow congruence:} The edge capacities $\flowLim{e}$ are large enough to sustain the flow through the graph. This is achieved by setting the capacity to be at least a fraction $\tau>1$ of the unconstrained Wardrop UE total flow $\SPTot$, which is in line with the literature~\cite{Motter2002}. 
\end{enumerate}
 
\subsubsection{Dynamics of congestion propagation} \label{emergencyPhase}
The previous section details the components of the flow network and the distribution of flow across it. In this section, we define the congestion dynamics triggered by an initial disruption. This disruption may cause congestion to propagate further through the network in multiple stages. Therefore, we use the superscript $(r)$ for every variable that is in stage $r\in \N$ of the cascade.

Stage $0$ represents the operation of the network before the first disruption occurs. Hence, the capacity vector $\flowLimVec^{(0)}$ is given by~\eqref{capacityPl} and the flow matrix $\flowM^{(0)}$ by $$\flowM^{(0)} = \flowM^*(\bm{U}, \flowLimVec^{(0)}).$$
The initial disruption can come from external factors such as the weather or internal factors such as a road accident. This disruption can lead to congestion and reduce the effective edge capacity, for example, when road lanes are closed. This affects the Wardrop UE, causing the redistribution of the network flow. Consequently, this may trigger additional disruptions throughout the network. Next, we describe this process in detail.

The cascade is initiated by a severe edge disruption at edge $e_1$, chosen uniformly at random from the set $\EE$. This reduces the effective capacity of edge~$e_1$ to 
$$\flowLim{e_1}^{(1)} = \phi_{init}\flowLim{e_1}^{(0)},$$ 
where $\phi_{init}  \in [l_m, l_M]\subset (0,1)$ is a continuous random variable with probability density function $p_{init}$ with support $[l_m, l_M]$. The capacities at all other edges remain the same. 

Figure~\ref{flowDiagram} illustrates the cascade process of disruptions. We now detail Steps 1 to 3 of Figure~\ref{flowDiagram} for an arbitrary cascade stage~$r\geq 2$: \begin{itemize}

    \item \textbf{Step 1.} To determine the set of disrupted edges, we consider the relative flow exceedance 
\begin{equation}
    \exc{e}^{(r)}:= \flow[e]^{(r)}/\flowLim{e}^{(r)}, \quad \text{for each } e\in \EE. \label{exceedance}
\end{equation}
Each overloaded edge $e$, i.e., with $\exc{e}^{(r)} > 1$,  experiences disruption with probability $p_e(\psi_e^{(r)})\in [0,1]$, independently of all other edge disruptions. Here, $p_e(\cdot)$ is some continuous, non-decreasing probability function. We denote the set of all disrupted edges in stage $r$ by $\EE^{(r)}$.
    
    \item As a consequence of the disruptions, in \textbf{Step 2}, the capacity of every edge $e\in \EE^{(r)}$ decreases. In particular, the updated edge capacity is now given by
\begin{equation} \label{nextFlowLim}
    \flowLim{e}^{(r)} := \begin{cases}
\phi_{e}\left(\exc{e}^{(r-1)}\right) \flowLim{e}^{(r-1)}, & \text{if } e\in \EE^{(r)},\\
\flowLim{e}^{(r-1)}, & \text{otherwise,}
        
    \end{cases} 
\end{equation}
where $\phi_e\in [l_m,l_M]\subset (0,1)$ is a continuous, non-decreasing function.

        \item In \textbf{Step 3}, the Wardrop UE flows are computed anew, using the updated flow capacity vector, i.e., $\flowM^{(r)} = \flowM^*(\bm{U},\flowLimVec^{(r)})$. 
            \end{itemize}
The cascade terminates once there are no edge disruptions in Step 1. For simplicity, we assume that each edge $e\in \EE$ can experience disruption at most $n_e\in \N$ times and we denote the current number of edge disruptions $e$ by $u_e$. If $u_e = n_e$, edge $e$ is excluded from the set of disrupted edges in Step~1. Thus, the cascade always terminates in a finite number of stages. We denote the random variable representing the disruption sequence by $D = (D_r)_{r\in \N}$, where $D_r$ is the set of edges disrupted in the $r$-th stage of the cascade. Furthermore, we denote the set of all possible disruption sequences by $\mathcal{D}$. We observe that $|\mathcal{D}|<\infty$ because every disruption sequence has finitely many stages and the graph has a finite number of edges.

To quantify the influence of the disruptions on the network performance, we consider a cascade measure $\Delta \flowCost^{(end)}$, capturing the increase in the flow cost due to the cascade. In particular, \begin{equation}
    \Delta\flowCost^{(end)} :=\flowCost(\flowM^{(end)}, \flowLimVec^{(end)}) -\flowCost(\flowM^{(0)}, \flowLimVec^{(0)}), \label{deltaFlowCost}
\end{equation}
where $\flowM^{(end)}$ and $\flowLimVec^{(end)}$ are the Wardrop UE flow and the flow capacity at the moment the disruption cascade terminates. Similarly, we define the added flow cost $\Delta \flowCost^{(r)}$ at the end of the $r$-th stage of the cascade as
\begin{equation}   \Delta\flowCost^{(r)} :=\flowCost(\flowM^{(r)}, \flowLimVec^{(r)}) -\flowCost(\flowM^{(0)}, \flowLimVec^{(0)}).\label{deltaFlowCostr}\end{equation}
Note that $\flowM^{(r)}$ and $\flowLimVec^{(r)}$ may be ill-defined because the cascade may have terminated at stage $q$, for some $q<r$. In such a case, we assume that $\flowM^{(r)} = \flowM^{(q)}$ and $\flowLimVec^{(r)} = \flowLimVec^{(q)}$.

Lastly, we note that the behavior of the model is probabilistic as it depends on the vertex weight vector $\pareto$, the initial capacity decrease factor $\phi_{init},$ and the disruption sequence $D$, which are random variables. Our analysis requires us to condition on these random variables. Hence, in this paper, for any quantity $A$, we use $A(\cdot)$ to denote $A$, conditioned on $(\cdot)$. The summary of the notation is provided in Table~\ref{tab:my_label} in Appendix~\ref{notation}.

\subsection{Analysis roadmap} \label{roadmap}
In this section, we discuss some key properties of the model that are crucial in proving the results we presented in Section~\ref{results}; namely, we show that the model exhibits scale-invariance, continuity, and the catastrophe principle. All these properties are formally proven in Appendix~\ref{analysis}.

In Appendix~\ref{subsec:scaleInv}, we show that our model is scale-invariant. This means that the behavior of the model is not dependent on the scale of vertex weights, but only on the proportions between them. This property is shown in steps through Lemmas~\ref{ScaleInvariancePl}--\ref{cascadeProbInv}.  Lemma~\ref{ScaleInvariancePl} shows the scale-invariance of the initial edge capacity vector. More specifically, for fixed vertex weight vector $\bm{X}$ and all $\omega>0$,
$$\flowLimVec^{(0)}(\omega\bm{X}) = \omega \flowLimVec^{(0)}(\bm{X}).$$ This follows from the properties of the Optimal Flow Problem~\eqref{optFlowProblem} and the definition of $\flowLimVec$ given in Equation~\eqref{capacityPl}. In Lemma~\ref{ScaleInvarianceCascade}, we show that the scale-invariance property can be extended to flow matrices and capacity vectors at every stage of the cascade, for a given vertex weight vector $\pareto$, initial capacity decrease factor $\phi_{init}$ and cascade $D$. In other words, for all $r\in \N$ and $\omega>0,$
\[\flowLimVec^{(r)}(\omega \pareto, \phi_{init}, D) = \omega \flowLimVec^{(r)}(\pareto, \phi_{init}, D) \quad \text{and}\quad\flowM^{(r)}(\omega \pareto, \phi_{init}, D) = \omega \flowM^{(r)}(\pareto, \phi_{init}, D). \]
We prove this using induction on the cascade stage, where the base case is again derived from the properties of the Optimal Flow Problem~\eqref{optFlowProblem}. This, together with the definition of the flow cost function and cascade cost (Equations~\eqref{costEquation}~and~\eqref{deltaFlowCostr}), directly implies that \[\Delta \flowCost^{(r)}(\omega \pareto, \phi_{init}, D) = \omega \Delta \flowCost^{(r)}(\pareto, \phi_{init}, D),\] which we show in Corollary~\ref{scaleInvDelta}. Last, in Lemma~\ref{cascadeProbInv}, we show that the probability of observing a particular disruption cascade $D$ is independent of $\omega$, i.e., 
\[\PR{D = d|\omega \bm{X},  \phi_{init}} =\PR{D = d|\bm{X},  \phi_{init}} \quad \text{for all } \omega>0.\]
This follows from the cascade mechanics and is proven using an iterative argument. An important consequence of these four results is that it is sufficient to study the cascade behavior for normalized vertex weight vectors $\bm{X}/\max(\bm{X})$, where the normalization is performed by dividing by the largest vertex weight. This explains the occurrence of normalized vertex weights in the Equation~\eqref{eq:Cr} of Theorem~\ref{mainTheorem}.

In Appendix~\ref{continuity}, we show that our model exhibits continuous behavior with respect to the vertex weight vector $\pareto$. This property indicates that a small change in the vertex weights has a small impact on the behavior of the network. In Lemmas~\ref{continuityPlanningOp} and~\ref{continuityFlowr}, we show that the continuity property holds for flow matrices and flow capacity vectors at every step of the cascade. Specifically, consider a convergent sequence of vertex weight vectors $(\bm{X}^k)_{k\in \N}$ with limit $\bm{X}^*\neq 0$, such that $\bm{X}^k\geq \bm{X}^*$ for all $k$. Further, take disruption sequence $D$ with initial congestion factor $\phi_{init}$. We then have that the flow capacity and the Wardrop UE flow at stage $r$, given $\pareto^k$, $\phi_{init}$ and $D$, converge to the flow capacity and the Wardrop UE flow at stage $r$, given $\pareto^*$, $\phi_{init}$, and $D$. In other words, \[\lim_{k\rightarrow \infty}\flowLimVec^{(r)}(\bm{X}^k,\phi_{init},D) = \flowLimVec^{(r)}(\bm{X}^*,\phi_{init},D)\quad \text{and}\quad
\lim_{k\rightarrow \infty}\flowM^{(r)}(\bm{X}^k,\phi_{init},D) = \flowM^{(r)}(\bm{X}^*,\phi_{init},D),\] for all $r\in \N$. This, together with the continuity of the flow cost function also implies the continuity of the congestion cost function at each step of the cascade. In Lemma~\ref{cascadeCont}, we show that the probability of a cascade is also a right-continuous function of vertex weights. 

The catastrophe principle, stated in Proposition~\ref{sbj}, indicates that with high probability cascades that lead to large congestion costs occur when one vertex has a much larger weight compared to all other vertices. We prove this proposition in Appendix~\ref{catastrophePrinc}, where we apply well-known properties of scale-free distributions and the fact that the flow cost at every stage of the cascade can be bounded by $M\cdot \max\{\pareto\}$ for some constant $M>0$. The latter follows from the scale-invariance of the flow cost function and Lemma~\ref{BoundFlow}. The catastrophe principle, together with the scale-invariance and continuity properties, implies that to understand the probability of large cascades, it is sufficient to analyze the cases where one vertex has weight 1 and all other vertices have weight 0, which explains the structure of the prefactor $C^{(r)}$ in Equation~\eqref{eq:Cr} of Theorem~\ref{mainTheorem} and the choice for our example in Section~\ref{results}.

Using these key properties of the model, we prove the main result, Theorem~\ref{mainTheorem}, in Appendix~\ref{mainTheoremProof}. The asymptotic behavior of the congestion cost $\Delta{c_f}$ is shown by constructing upper and lower bounds for the congestion cost probabilities at every stage of the cascade, and showing that both bounds are asymptotically equal as $y\rightarrow \infty$. 
\section*{Acknowledgments}
This work is supported by NWO through Gravitation NETWORKS grant no. 024.002.003.

\bibliographystyle{naturemag}
\newpage \bibliography{sample}
\appendix
\newpage
\section{Notation} \label{notation}
\begin{table}[h!]
    \centering
    \renewcommand{\arraystretch}{1.16}
    \begin{tabular}{|c|l|}
    \hline
    \textbf{Variable} & \textbf{Description}\\
    \hline
    $\N$ & The set of natural numbers excluding 0\\
    \hline
    $\R_+$ & The set of non-negative real numbers \\
    \hline
         $\G = (\V,\EE)$ & A connected, directed graph representing the underlying network\\
    \hline
         $\V$ & The set of vertices of $\G$\\
    \hline
         $\EE$ & The set of edges of $\G$ \\
    \hline
         $\nv = |\V|$ & Number of vertices in the network\\
    \hline
         $\nee = |\EE|$ & Number of edges in the network\\
    \hline
         $\incM$ & The incidence matrix of graph $\G$ \\
    \hline
         $\bm{X} = (X_v)_{v\in \V}$ & The vertex weight vector, with $X_v$ being the weight of vertex $v\in \V$\\
    \hline
         $\text{diag}(\pareto)$ & Diagonal matrix of vertex weights \\
    \hline
         $q_{v,w}\geq 0$ & Fraction of weight of vertex $w$ with destination at vertex $v$ \\
    \hline
        $\gen[v,w]\in \R_+$ & Amount of commodity $w\in \V$ with origin at vertex $v\in \V$\\
    \hline
         $\dem[v,w] := q_{v,w}\paretoI{w}\in \R_+$ & Amount commodity $w\in\V$ with destination vertex $v\in \V$\\
    \hline

         $\bm{Q} := [q_{v,w}]_{v,w\in \V}$ & Travel factor matrix of size $\nv\times \nv$\\
    \hline

         $\genM := [\gen[v,w]]_{v,w\in \V}$ & Origin matrix of size $\nv\times \nv$, with $\genM = \text{diag}(\pareto)$\\
         
    \hline
             $\demM := [\dem[v,w]]_{v,w\in \V}$ & Destination matrix of size $\nv\times \nv$, with $\demM = \bm{Q}\genM$\\
    \hline
         $\bm{U} := \demM- \genM$ & Net travel matrix of size $\nv\times \nv$\\
    \hline
         $\flow[e,w]$ & Flow of commodity $w\in \V$ on edge $e\in \EE$ with $\flow[e,w]\geq 0$\\
    \hline
         $\flowM := [\flow[e,w]]_{e\in\EE,w\in \V}$ & Flow matrix of size $\nee\times \nv$\\
    \hline
         $\bm{e}$ & All-ones vector of appropriate size\\
    \hline
         $\unitV{j}$ & $j$-th unit vector of size $\nv$\\
    \hline
         $\flow[e] := (\flowM\bm{e})_e$ & Total flow on edge $e\in \EE$, i.e., the sum of flows of commodities on edge $e$ \\
    \hline
         $\flowTot:= (\flow[e])_{e\in\EE}$ & Vector of total edge flows\\
    \hline
         $\flowLim{e}$ & The capacity of edge $e\in \EE$\\
    \hline
         $\flowLimVec :=(\flowLim{e})_{e\in\EE}$ & Vector of edge capacities \\
    \hline
         $\flowCost(\flowM, \flowLimVec)$ & Flow cost function, $\flowCost: \R^{\nee\times\nk}_+\times \R^{\nee}_+ \rightarrow \R_+$ \\
    \hline
        $\flowMOpt{\bm{U}, \flowLimVec}$ & Function describing the flow distribution mechanism (see Problem~\eqref{optFlowProblem})\\
    \hline
        $\SPM(.)$ & The shortest path flow matrix, given by $\SPM(.) = \flowM^*(., \infty)$\\
            \hline
    $\SPTot(.)$ & The shortest path total flow vector, given by $\SPTot(.) = \SPM(.) \bm{e}$\\
    \hline
        $\tau$ & Planning slack parameter; $\tau\geq 1$ \\
    \hline
        $\phi_{init}$ & The initial (random) congestion factor\\
    \hline 
        $p_{init}$ & The probability mass function of $\phi_{init}$\\
    \hline 
        $\exc{e}$ & Relative flow exceedance on edge $e\in \EE$; see Equation~\eqref{exceedance}\\
    \hline 
        $n_e$ & Number of disruption states of edge $e\in \EE$ \\
    \hline
        $u_e$ & Disruption state of edge $e\in \EE$; $u_e\in\{0,1,\dots, n_e\}$\\
    \hline 
        $\phi_e(\exc)$ & Congestion factor of edge $e\in \EE$; $\phi_e: \R_+\rightarrow [l_m,l_M]\subset (0,1)$\\
    \hline 
        $p_e(\psi)$ & Disruption probability of edge $e\in \EE$; $p_e: \R_+\rightarrow [0,1]$\\
    \hline 
        $\Delta \flowCost$ & Congestion cost, see Equations~\eqref{deltaFlowCost} and~\eqref{deltaFlowCostr}\\
    \hline
        $D$ & Disruption sequence\\
    \hline 
        $\mathcal{D}$ & Set of all possible disruption sequences $D$ \\
    \hline
    \end{tabular}
    \renewcommand{\arraystretch}{1}
    \caption{Summary of notation.} 
    \label{tab:my_label}
\end{table}

\newpage\section{Scale-free characteristics in Dutch traffic data}\label{DataAnalysis}

To propose a meaningful theoretical framework for traffic congestion, it is essential to first understand the fundamental characteristics of real traffic networks. Recent analyses of congestion data from China have suggested the existence of scale-free behavior in their transportation system \cite{ Chen_Lin_Yan_Liu_Liu_Li_2024, Zhang_Zeng_Li_Huang_Stanley_Havlin_2019}. However, to the best of our knowledge, similar results have not been documented for other transportation networks or other traffic measures, such as traffic intensity. Therefore, in this section, we analyze empirical data from the highway network in the Netherlands, aiming to estimate the distributions of traffic intensity and congestion. 

In this analysis, we utilize two datasets. \textit{Dataset 1} consists of Origin-Destination (OD) matrices, representing the daily number of trips between each OD  pair. Each location (origin or destination) represents one of 355 Dutch municipalities (as of 2020) or the external location. The dataset spans the period from January 2020 to December 2022. \textit{Dataset 2} comprises traffic jam records on Dutch highways, collected between January 2018 and August 2024. Each traffic jam record is represented as a time series, showing the evolution of the traffic jam, measured in meters. This dataset is sourced from the Dutch open data traffic platform NDW. 

The traffic intensity distribution is estimated using \textit{Dataset 1}. Specifically, we compute the average daily OD matrix, denoted as $\bm{M_{AOD}}$, by taking the mean of all OD matrices in the dataset. The traffic intensity $(v_{TI})_j$ from location $j\in \{1,\dots, 356\}$ is given by aggregating the number of trips from location $j$ over all possible destinations $i\in \{1,\dots, 356\}$. In particular,
\[\bm{v_{TI}} = \bm{M_{AOD}} \cdot \bm{e},\]
where $\bm{e}$ is the all-ones vector of size 356. This municipal traffic intensity is a reasonable proxy because the Dutch highway network is dense and highly developed, and it supports most inter-municipality trips. 

Figure~\ref{fig:a} shows the log-log plots of the empirical Complementary Cumulative Distribution Function (CCDF) for the traffic intensity. The tail of the distribution exhibits an approximately linear behavior, suggesting scale-free distributional properties. To explore this further, we estimate the cutoff point $\kappa_{TI} = 197$ (Figure~\ref{fig:e}), i.e., the threshold where the scale-free tail begins, using the Power-Law FIT (PLFIT) method \cite{plfit09}. Then, we use the Hill estimator $\xi_{TI}$ \cite{nair_wierman_zwart_2022} to find the appropriate tail parameter, given as the reciprocal of $\xi_{TI}$ at $\kappa_{TI}$ (Figure~\ref{fig:c}). Lastly, we perform the Kolmogorov-Smirnoff (KS) test, to determine if the power law hypothesis can be rejected at the 95\% confidence level. The recovered tail parameter $\alpha_{TI} = 1.42$. Table~\ref{tab:KS} summarizes the results of the KS test, showing that the power-law hypothesis cannot be rejected. However, due to the relatively small sample size, the strength of the test is limited. Based on this evidence, we conclude that there is a premise that the traffic intensity distribution in the Netherlands is Pareto-tailed. Nonetheless, further statistical analysis is required to confirm this claim.

Next, we analyze the distribution of congestion. For each unique jam, the congestion estimator is the average traffic jam length over time. Again, in the log-log plot of the empirical CCDF shown in Figure~\ref{fig:b}, we observe a linear tail behavior, suggesting a Pareto-tailed phenomenon. Hence, we perform the PLFIT method with the Hill estimator to fit the appropriate tail parameters (Figure~\ref{fig:e}). Then, we verify the Pareto-tailed hypothesis with the KS test. 

We obtain the scale parameter $\alpha_C = 6.72$ for the congestion distribution. For this parameter, in Figure~\ref{fig:d}, we observe that the corresponding Hill plot flattens, which is the expected behavior when the data is Pareto-tailed. Again, according to the KS test results, shown in Table~\ref{tab:KS}, the Pareto hypothesis cannot be rejected at the 95\% confidence level. Note that the strength of this test is larger than the two previous tests as a result of the larger sample size. 

The performed analysis suggests that the Dutch highway transportation system exhibits scale-free characteristics in both traffic intensity and congestion. The latter finding aligns with the results of \cite{Chen_Lin_Yan_Liu_Liu_Li_2024} and \cite{Zhang_Zeng_Li_Huang_Stanley_Havlin_2019}, where the authors found evidence of scale-free congestion in Chinese traffic systems. 

\begin{table}[h!]
    \centering
    \begin{tabular}{|c|c|c|c|c|}
    \hline
         \textbf{Data} & \textbf{Test statistic} & \textbf{p-value} & \textbf{\# samples} & \textbf{Reject} \\
         \hline 
         \textbf{Traffic intensity} & 0.09 & 0.35 & 197 & No\\
         \hline 
         \textbf{Traffic jams} & 0.04 & 0.34 & 1049 & No\\
         \hline
    \end{tabular}
    \caption{KS test results.}
    \label{tab:KS}
\end{table}

\begin{figure}[p!] 
\begin{subfigure}{0.45\textwidth}
\includegraphics[width=\linewidth]{figures/Asset_2.pdf}
\caption{CCDF plot of traffic intensity distribution} \label{fig:a}
\end{subfigure}\hspace*{\fill}
\begin{subfigure}{0.45\textwidth}
\includegraphics[width=\linewidth]{figures/Asset_1.pdf}
\caption{CCDF plot of congestion distribution} \label{fig:b}
\end{subfigure}

\begin{subfigure}{0.45\textwidth}
\includegraphics[width=\linewidth]{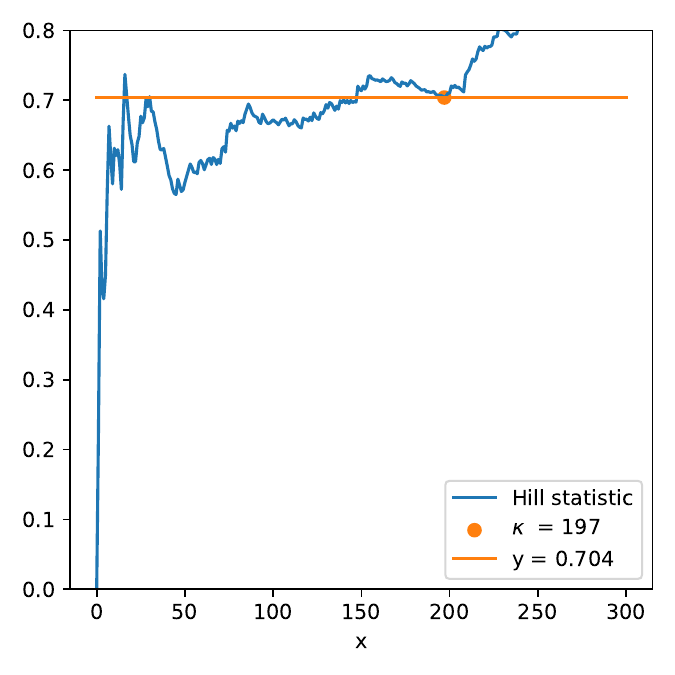}
\caption{Hill plot for traffic intensity} \label{fig:c}
\end{subfigure}\hspace*{\fill}
\begin{subfigure}{0.45\textwidth}
\includegraphics[width=\linewidth]{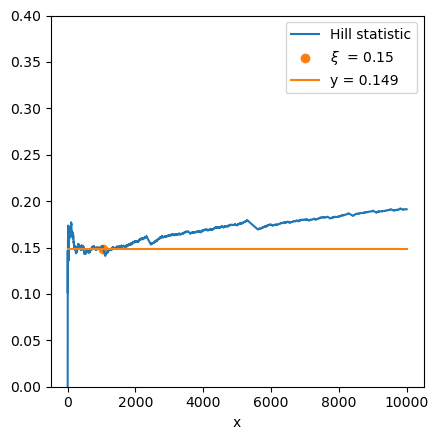}
\caption{Hill plot for congestion} \label{fig:d}
\end{subfigure}

\begin{subfigure}{0.45\textwidth}
\includegraphics[width=\linewidth]{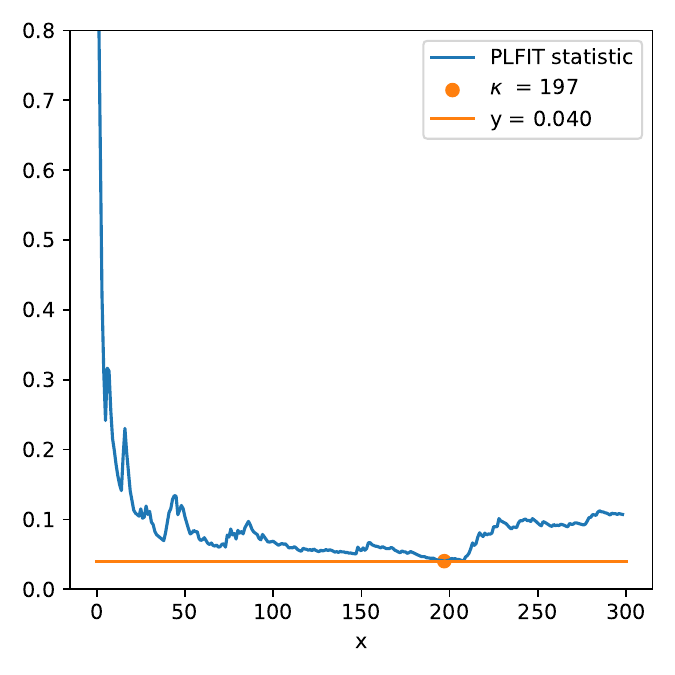}
\caption{PLFIT plot for traffic intensity} \label{fig:e}
\end{subfigure}
\hspace*{\fill}
\begin{subfigure}{0.45\textwidth}
\includegraphics[width=\linewidth]{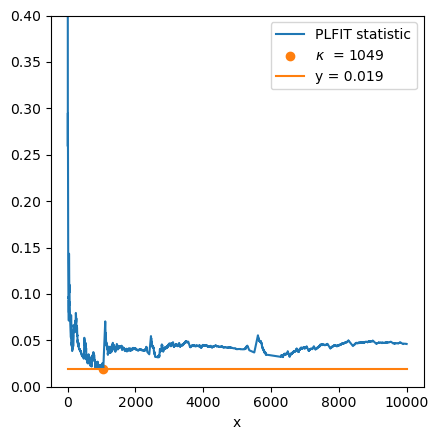}
\caption{PLFIT plot for congestion} \label{fig:f}
\end{subfigure}%
\caption{Tail analyses of traffic intensity distribution (left) and congestion distribution (right).} \label{fig:1}
\end{figure}
\newpage\section{Example}\label{example}
In this section, we provide the details of the cascade simulation example. Table~\ref{tabExample} specifies all model parameters used in the simulation example. Note that the choice of the free-flow travel parameter $\bm{d}$ of the flow cost function ensures that Problem~\eqref{optFlowProblem} with input $\bm{U} = \bm{Q}\cdot \text{diag}(\bm{X}) - \text{diag}(\bm{X})$ and $\flowLimVec = \infty$ always has a unique solution. Moreover, since the upper and lower bounds of the initial congestion factor are $l_m = l_M = 1/20$, we know that $\phi_{init} = \phi_e = 1/20$, and the initial probability $p_{init}(x) = 1$ for $x = 1/20$ and 0 otherwise.

Figure~\ref{cascade} presents the results on cascade progression in a single simulation run with initial vertex weights $\bm{X} = (1,0,0,0,0)^T$. Figures~\ref{loglogdist} --~\ref{HillPlot1} show results on the tail of the cascade cost distribution. These results were generated with $N = 10^6$ simulation runs where $X_v$ follows a Pareto distribution on $(1,\infty)$ with parameter $\alpha = 3/2$. 

\setcounter{MaxMatrixCols}{14}
\begin{table}[ht!]
    \centering
    \renewcommand{\arraystretch}{1.1}
    \begin{tabular}{|c|c|}
    \hline
        \textbf{Variable} & \textbf{Value}  \\
        \hline
        $\V$ & $\{1,2,3,4,5\}$\\
        \hline
        $\EE$ & $\{(1,2),(1,3),(1,4),(2,1),(2,3),(2,5),(3,1),(3,2),(3,4),(4,1),(4,3),(4,5),(5,2),(5,4)\}$\\
        \hline
        $\nv$ & 5\\
        \hline
        $\nee$ & 14\\
        \hline 
        $\incM$ & $\begin{pmatrix}
            -1 & -1 & -1 & 1 & 0 & 0 & 1 & 0 & 0 & 1 & 0 & 0 & 0 & 0 \\
            1 & 0 & 0 & -1 & -1 & -1 & 0 & 1 & 0 & 0 & 0 & 0 & 1 & 0\\
            0 & 1 & 0 & 0 & 1 & 0 & -1 & -1 & -1 & 0 & 1 & 0 & 0 & 0\\
            0 & 0 & 1 & 0 & 0 & 0 & 0 & 0 & 1 & -1 & -1 & -1 & 0 & 1\\
            0 & 0 & 0 & 0 & 0 & 1 & 0 & 0 & 0 & 0 & 0 & 1 & -1 & -1 
        \end{pmatrix}$\\
        \hline 
        $\tau $ & 51/50\\
        \hline
        $\varepsilon_{\min}$ & 1/10\\
        \hline
        $\bm{d} = (d_e)_{e\in \EE}$ & $(1, 1, 1, 1, 3/2, 2, 1, 3/2, 1, 1, 1, 1, 2, 1)^T$
        \\
        \hline
        $\bm{b} = (b_e)_{e\in \EE}$& $(1,1,1,1,1,1,1,1,1,1,1,1,1,1)^T$\\
        \hline
        $\beta$ & 1\\
        \hline 
        $\alpha$ & 3/2\\
        \hline
        $\bm{Q}$ & $\begin{pmatrix}
      0 & 1/8 &  0 &1/4 & 1/3\\
     1/2&  0& 0& 1/4& 1/6\\
     1/4& 1/2& 0& 1/4& 1/6\\
     1/8& 1/8& 1& 0& 1/3\\
     1/8& 1/4& 0& 1/4& 0
        \end{pmatrix}$\\
        \hline
        $l_m$ & 1/20\\
        \hline 
        $l_M$ & 1/20\\
        \hline 
        $p_e$ & $\begin{cases}
            0 & \text{if } \phi_e^{(r)}\leq 1,\\
            \phi_e^{(r)} - 1 & \text{if } \phi_e^{(r)}\in (1,2),\\
            1 & \text{if } \phi_e^{(r)} \geq 2,
        \end{cases}$ \quad \quad (for all $e\in \EE$ and $r\in \N$)\\
        \hline 
        $n_e$ & 1 for all $e\in \EE$\\
        \hline

    \end{tabular}
    \caption{Model parameters used in the simulation.}
    \label{tabExample}
\end{table}

\newpage\section{Analysis}\label{analysis}

In this section, we state and prove results about the key properties of the model: scale-invariance, continuity, and the catastrophe principle. Then, we provide the proof of Theorem~\ref{mainTheorem}. Lastly, we show the strict convexity of the flow cost function and discuss the properties of the optimal flow function with a linear cost function. Note that for all results, we assume that the graph $\G = (\V,\EE)$ and the travel factor matrix $\bm{Q}$ are fixed. 
\subsection{Scale-invariance} \label{subsec:scaleInv}In this section, we present four results that show the scaling properties of flow matrices, flow capacity vectors, flow exceedance vectors, congestion cost, and the probability of a cascade sequence to occur. Using these properties, we can show that it is sufficient to analyze the cascade behavior for a graph with a normalized set of vertex weights. First, we show the scale-invariance property of the capacity vector $\flowLimVec^{(0)}$. 
\begin{lemma}[Scale-invariance of $\flowLimVec^{(0)}$]
Consider a fixed vector of vertex weights $\bm{X} = (\paretoI{1},\dots,\paretoI{\nv})$. Then, for all $\omega>0$,
$$\flowLimVec^{(0)}(\omega\bm{X}) = \omega \flowLimVec^{(0)}(\bm{X}),$$
where $\flowLimVec^{(0)}(.)$, denotes the conditional initial flow capacity vector. 
\label{ScaleInvariancePl}

\end{lemma}

\begin{proof}
Let $\bm{X} = (\paretoI{1}, \dots, \paretoI{\nv})$ be a given vector of vertex weights. We have that  $\genM = \text{diag}(\pareto),$ $\demM = \bm{Q}\genM$, and $\bm{U} = (\bm{Q} - \bm{I}_{\nv})\genM$, where $\bm{I}_{\nv}$ is the $\nv\times\nv$ identity matrix. From Equation~\eqref{capacityPl}, we know that for every $(v,w)\in \EE$,
    \begin{equation}
        \flowLim{(v,w)}^{(0)}(\bm{X}) =\max\left\{\tau g_{(v,w)}, \tau g_{(w,v)}, \varepsilon_{\min}\sum_{i = 1}^{\nv}X_i\right\} \label{fBarpl}.\end{equation}

Next, we show that $\SPM(\omega \bm{U}) = \flowM^*(\omega \bm{U}, \infty) = \omega\flowM^*(\bm{U},\infty) = \omega \SPM(\bm{U})$. First, we show that $\omega\SPM(\bm{U})$ is a feasible solution to Problem~\eqref{optFlowProblem} with input $\omega\bm{U}$ and $\infty$. By optimality of $\SPM(\bm{U})$, we have that 
\[\incM \SPM(\bm{U}) = \bm{U},\]
which implies that $\incM\left( \omega\SPM(\bm{U}) \right)= \omega\bm{U}.$

It remains to show that it is also the optimal solution. By the optimality of $\SPM(\omega \bm{U})$ and feasibility of  $\omega \SPM(\bm{U})$, we know that 
\[c_f(\SPM(\omega\bm{U}),\infty)\leq c_f\left(\omega \SPM(\bm{U}), \infty \right).\]
Using Equation~\eqref{costEquation}, we find
\[\sum_{e\in \EE} d_e\left(\SPTot(\omega \bm{U})\right)_e\leq\sum_{e\in \EE} d_e\left(\omega \SPTot(\bm{U})\right)_e, \]
and therefore,
\begin{equation*}\sum_{e\in \EE} d_e\left(\frac{1}{\omega}\SPTot(\omega\bm{U})\right)_e \leq \sum_{e\in \EE} d_e\left(\SPTot(\bm{U})\right)_e.\end{equation*}
We notice that $\frac{1}{\omega}\SPTot(\omega\bm{U})$ is a feasible solution to Problem~\eqref{optFlowProblem} with input $\bm{U}$ and $\infty$. Moreover, its cost is lower or equal to the cost of $\SPM(\bm{U})$. Hence, by optimality of $\SPM(\bm{U})$ we conclude that the costs of $\SPM(\bm{U})$ and $\frac{1}{\omega}\SPM(\omega \bm{U})$ are equal. Note that this also implies that if for some $\omega^*>0$, Problem~\eqref{optFlowProblem} with input $(\omega^*\bm{U}, \infty)$ has multiple optimal solutions, then the same problem with input $(\omega\bm{U}, \infty)$ has multiple optimal solutions for all $\omega >0$ because we can construct optimal solutions corresponding to $\omega$ from the optimal solutions corresponding $\omega^*$ by multiplying them with a factor $\omega^*/\omega$. Thus, if Problem~\eqref{optFlowProblem} with input $(\bm{U}, \infty)$ attains a unique optimal solution, then 
\begin{equation}
\SPM(\omega \bm{U}) = \flowM^*(\omega\bm{U},\infty) = \omega\flowM^*(\bm{U},\infty) = \omega\SPM(\bm{U}).  \label{flowMeq}  
\end{equation}
For the non-unique case, we first observe that the set of all shortest paths between vertices $v$ and $w$ for all $v,w\in \V$ for which $u_{v,w}\neq 0$ is the same for all $\omega \bm{U}$, $\omega>0$. Moreover, flow $\SPM(\bm{U})$ is given by the solution of Problem~\eqref{optFlowProblem} that assigns equal flow to all shortest paths between vertices $v$ and $w$ for all $v,w\in \V$ for which $u_{v,w}\neq 0$. But then $\omega\SPM(\bm{U})$ must also assign equal flow to all shortest paths and, since it is a feasible and optimal solution to Problem~\eqref{optFlowProblem} with input $(\omega\bm{U}, \infty)$, it follows that $\SPM(\omega \bm{U}) = \omega\SPM(\bm{U})$. 

Lastly, using Equations~\eqref{capacityPl},~\eqref{fBarpl}, and~\eqref{flowMeq}, we conclude that $\omega \flowLim{v,w}^{(0)}(\bm{X}) = \flowLim{v,w}^{(0)}(\omega\bm{X})$ for all edges $(v,w)\in \EE$.

\end{proof}
Next, we show that the scale-invariance property also holds for flow matrices and flow capacity vector at every stage of a given cascade $D = d$. Note that the probability that a given cascade occurs is not dependent on $\omega$, which we show later on in Lemma~\ref{cascadeProbInv}. Before stating the result, we introduce some notation. For $\omega>0$, let $\flowM(\omega\bm{X},\phi_{init}, D) = \left(\flowM^{(r)}(\omega\bm{X},\phi_{init}, D)\right)_{r\in \{0,\dots, |D|\}}$, $\flowLimVec(\omega\bm{X},\phi_{init},D) = \left(\flowLimVec^{(r)}(\omega\bm{X},\phi_{init},D)\right)_{r\in \{0,\dots, |D|\}}$, and $\bm{\psi}(\omega\bm{X},\phi_{init}, D) = \left(\bm{\psi}^{(0)}(\omega\bm{X},\phi_{init}, D)\right)_{r\in \{0,\dots, |D|\}}$ denote the sequences of flow matrices, flow capacity vectors, and exceedance vectors at every cascade stage, given the set of vertex weights $\omega\bm{X}$, the cascade sequence $D$, and the initial capacity decrease factor $\phi_{init}$. 

\begin{lemma}
    Consider a fixed vector of vertex weights $\bm{X} = (\paretoI{1},\dots,\paretoI{\nv})$, a sequence of edge disruptions $D = (D_r)_{r\in \{1,\dots, |D|\}}$ and an initial capacity decrease factor $\phi_{init}$. 
    Then, for all $\omega>0$,\begin{enumerate}
        \item $\flowM(\omega\bm{X},\phi_{init},D) = \omega\flowM(\bm{X},\phi_{init}, D),$
        \item $\flowLimVec(\omega\bm{X},\phi_{init},D) = \omega\flowLimVec(\bm{X},\phi_{init},D),$
\item $\bm{\psi}(\omega\bm{X},\phi_{init}, D) = \bm{\psi}(\bm{X},\phi_{init}, D).$

    \end{enumerate}
    \label{ScaleInvarianceCascade}
\end{lemma}
We show this by induction.
\begin{proof}
    First, applying Lemma~\ref{ScaleInvariancePl}, we obtain that $$\flowLimVec^{(0)}(\omega\bm{X},\phi_{init},D) = \omega \flowLimVec^{(0)}(\bm{X},\phi_{init},D)$$
    as the initial capacity vector depends only on vertex weights and not on the cascade $D$ and the factor $\phi_{init}$. Next, we show that $\flowM^{(0)}(\omega\bm{X},\phi_{init},D) = \omega\flowM^{(0)}(\bm{X},\phi_{init},D)$, using a similar argument as in the proof of Lemma~\ref{ScaleInvariancePl}. 

    Recall that $\bm{U} = (\bm{Q} - \bm{I}_{\nv})\cdot \text{diag}(\pareto)$ and $\flowM^{(0)}(\omega\bm{X},\phi_{init},D) = \flowM^*\left(\omega\bm{U}, \omega \flowLimVec^{(0)}(\bm{X},\phi_{init},D)\right)$. First, we observe that $\omega \flowM^{(0)}(\bm{X},\phi_{init},D)$ is a feasible solution to Problem~\eqref{optFlowProblem} with net travel matrix $\omega \bm{U}$. It remains to show by contradiction that it is the optimal solution. Suppose that $\flowM^*\left(\omega\bm{U}, \omega \flowLimVec^{(0)}(\bm{X},\phi_{init},D)\right) = \widetilde{\flowM}\neq \omega \flowM^{(0)}(\bm{X},\phi_{init},D)$. This implies that 
    \begin{equation}\label{costCont}c_f\left(\widetilde{\flowM},\omega \flowLimVec^{(0)}(\bm{X},\phi_{init},D)\right)<c_f\left(\omega \flowM^{(0)}(\bm{X},\phi_{init},D), \omega\flowLimVec^{(0)}(\bm{X},\phi_{init},D) \right).\end{equation}
    Let $\tilde{\flowTot} = \widetilde{\flowM}\bm{e}$ and $\flowTot^{(0)} = \flowM^{(0)}(\bm{X},\flowLimVec^{(0)}(\bm{X},\phi_{init},D))\bm{e}$. Applying the definition of the cost function and dividing both sides of the inequality~\eqref{costCont} by $\omega$, we obtain

      \begin{align*}
&\sum_{e\in \EE}\frac{d_e}{\omega}\tilde{f}_e + \frac{b_e}{\beta+1}\flowLim{e}^{(0)}(\bm{X},\phi_{init},D) \left(\frac{1}{\omega}\tilde{f}_e/ \flowLim{e}^{(0)}(\bm{X},\phi_{init},D) \right)^{\beta+1}\\
&\hspace{0.7cm}< \sum_{e\in \EE} d_ef^{(0)}_e +\frac{b_e}{\beta+1}\flowLim{e}^{(0)}(\bm{X},\phi_{init},D) \left(f_e^{(0)}/ \flowLim{e}^{(0)}(\bm{X},\phi_{init},D) \right)^{\beta+1}.    
    \end{align*}  
However, this implies that $\frac{1}{\omega} \widetilde\flowM$ has a smaller cost than $\flowM^{(0)}(\bm{X},\phi_{init},D)$, which is a contradiction. Hence, we conclude that 
\begin{equation}\flowM^{(0)}(\omega\bm{X},\phi_{init},D) = \omega\flowM^{(0)}(\bm{X},\phi_{init},D).
\label{contradictionArg}\end{equation}
Furthermore, we find that for every $e\in \EE$
\[\exc{e}^{(0)}(\omega\bm{X},\phi_{init}, D)  = \frac{\omega f_e^{(0)}}{\omega \flowLim{e}^{(0)}(\bm{X},\phi_{init},D)} = \exc{e}^{(0)}(\bm{X},\phi_{init}, D).\]
Next, we show by induction that $\flowM^{(r)}(\omega\bm{X},\phi_{init},D) = \omega\flowM^{(r)}(\bm{X},\phi_{init},D)$ and $\flowLimVec^{(r)}(\omega\bm{X},\phi_{init},D) = \omega \flowLimVec^{(r)}(\bm{X},\phi_{init},D)$ for any $r\in \{1,\dots, |D|\}$. Suppose that it is true up to $r\in \{0,\dots, |D| -1 \}$.  Then, for every $e\in \EE$, \begin{align*}&\exc{e}^{(r)}(\omega\bm{X},\phi_{init},D)=\flow[e]^{(r)}(\omega\bm{X},\phi_{init},D)/\flowLim{e}^{(r)}(\omega\bm{X},\phi_{init},D) \\&\hspace{0.7cm}=\flow[e]^{(r)}(\bm{X},\phi_{init},D)/\flowLim{e}^{(r)}(\bm{X},\phi_{init},D) = \exc{e}^{(r)}(\bm{X},\phi_{init},D).\end{align*}
Hence, by~\eqref{nextFlowLim}, we conclude that 
\[\flowLimVec^{(r+1)}(\omega\bm{X},\phi_{init},D) = \omega \flowLimVec^{(r+1)}(\bm{X},\phi_{init},D).\]
With this we can also show that $\flowM^{(r+1)}(\omega\bm{X},\phi_{init},D) = \omega\flowM^{(r+1)}(\bm{X},\phi_{init},D)$, which follows from an analogue contradiction argument used to show~\eqref{contradictionArg}, by replacing superscript $(0)$ with $(r)$.
\end{proof}

From Lemmas~\ref{ScaleInvariancePl} and~\ref{ScaleInvarianceCascade}, the next result follows directly.
\begin{corollary}
    Consider vectors of vertex weights $\bm{X} = (\paretoI{1},\dots,\paretoI{\nv})
    $ initial congestion faction $\phi_{init}$, and a sequence of edge disruptions $D = (D_r)_{r\in \{1,\dots, |D|\}}$. Let $\omega>0$ and consider $r\in\{1,\dots,|D|\}$.
    Then, \begin{enumerate}
    \item $c_f\left(\flowM^{(k)}(\omega\bm{X}, \phi_{init}, D), \flowLimVec^{(k)}(\omega\bm{X}, \phi_{init}, D)\right) = \omega c_f\left(\flowM^{(k)}(\bm{X}, \phi_{init}, D), \flowLimVec^{(k)}(\bm{X}, \phi_{init}, D)\right),$ \item
$\Delta c_f^{(r)}(\omega\bm{X}, \phi_{init}, D) = \omega\Delta c_f^{(r)}(\bm{X}, \phi_{init}, D).$ \end{enumerate}\label{scaleInvDelta}
\end{corollary}
\begin{proof}
    From Lemmas~\ref{ScaleInvariancePl} and ~\ref{ScaleInvarianceCascade}, we obtain that $\flowM^{(k)}(\omega\bm{X}, \phi_{init}, D) = \omega\flowM^{(k)}(\bm{X}, \phi_{init}, D) $ and $\flowLimVec^{(k)}(\omega\bm{X}, \phi_{init}, D) = \omega\flowLimVec^{(k)}(\bm{X}, \phi_{init}, D),$ for all $k\in \{0,\dots,|D|\}.$ With this, we calculate the cost at the $k$-th stage of the cascade
    \begin{align*}
        &c_f\left(\flowM^{(k)}(\omega\bm{X}, \phi_{init}, D), \flowLimVec^{(k)}(\omega\bm{X}, \phi_{init}, D)\right)\\
        &\hspace{0.7cm}= \sum_{e\in \EE} \left(d_e\flow[e]^{(k)}(\omega\bm{X}, \phi_{init}, D) + \frac{b_e}{\beta +1} \flowLim{e}^{(k)}(\omega\bm{X}, \phi_{init}, D)\left(\frac{\flow[e]^{(k)}(\omega\bm{X}, \phi_{init}, D)}{ \flowLim{e}^{(k)}(\omega\bm{X}, \phi_{init}, D)}\right)^{\beta+1}\right)\\
        &\hspace{0.7cm}=\omega \sum_{e\in \EE} \left(d_e\flow[e]^{(k)}(\bm{X}, \phi_{init}, D) + \frac{b_e}{\beta +1} \flowLim{e}^{(k)}(\bm{X}, \phi_{init}, D)\left(\frac{\flow[e]^{(k)}(\bm{X}, \phi_{init}, D)}{ \flowLim{e}^{(k)}(\bm{X}, \phi_{init}, D)}\right)^{\beta+1}\right)\\
        &\hspace{0.7cm}= \omega c_f\left(\flowM^{(k)}(\bm{X}, \phi_{init}, D), \flowLimVec^{(k)}(\bm{X}, \phi_{init}, D)\right).
    \end{align*}
    Hence, \begin{align*}&\Delta c_f^{(r)}(\omega\bm{X}, \phi_{init}, D) \\&\hspace{0.7cm}= c_f\left(\flowM^{(r)}(\omega\bm{X}, \phi_{init}, D), \flowLimVec^{(r)}(\omega\bm{X}, \phi_{init}, D)\right) - c_f\left(\flowM^{(0)}(\omega\bm{X}, \phi_{init}, D), \flowLimVec^{(0)}(\omega\bm{X}, \phi_{init}, D)\right)\\&\hspace{0.7cm} =  \omega c_f\left(\flowM^{(r)}(\bm{X}, \phi_{init}, D), \flowLimVec^{(r)}(\bm{X}, \phi_{init}, D)\right) -  \omega c_f\left(\flowM^{(0)}(\bm{X}, \phi_{init}, D), \flowLimVec^{(0)}(\bm{X}, \phi_{init}, D)\right) \\&\hspace{0.7cm}= \omega \Delta c_f^{(r)}(\bm{X}, \phi_{init}, D).\end{align*}
\end{proof}
Finally, we present and prove a lemma, showing that the probability of a certain cascade to occur is independent of the scale~$\omega$. 
\begin{lemma}
     Consider vectors of vertex weights $\bm{X} = (\paretoI{1},\dots,\paretoI{\nv})$ and let $\PR{D = d | \bm{X}, \phi_{init}}$ be the probability that the cascade sequence  $d = (d_r)_{r\in \{1,\dots, |d|\}}$ occurs, given the vector of vertex weights $\bm{X}$ and the initial capacity decrease factor  $\phi_{init}$.  Then,
     \[\PR{D = d|\omega \bm{X},  \phi_{init}} =\PR{D = d|\bm{X},  \phi_{init}} \quad \text{for all } \omega>0.\]
In other words, this probability does not depend on the scale of vertex weights.
\label{cascadeProbInv}\end{lemma}
\begin{proof}

Let $A(d_r)$ be the event that edges in set $d_r$ become disrupted in the cascade stage $r$.
Observe that through an iterative argument, it suffices to show the following. Given a set of disruptions so far, the probability that an arbitrary edge $l\in \EE$ is disrupted (fails) in the next cascade stage is equal for all $\omega >0$. More precisely, we will show that for all $l\in \EE$, $r\in \{2,\dots, |D|\}$ and any sequence of disruptions $A(d_1), \dots A(d_{r-1}),$
\begin{equation}\PR{l \text{ fails in stage } r|\omega \bm{X}, \phi_{init}, A(d_1), \dots, A(d_{r-1})} = \PR{l \text{ fails in stage } r| \bm{X}, \phi_{init}, A(d_1), \dots, A(d_{r-1})}.\label{probLFails}\end{equation}
In our model, the probability that edge $l$ fails in stage $r$ is given by $p_l(\exc{l}^{(r)})$, independently of other edge disruptions, but the edge can only fail if it is not in its final disruption state, i.e., $u_l^{(r)}\neq n_l$. Hence, we find that:
\begin{align*}
    &\PR{l \text{ fails in stage } r|\omega \bm{X}, \phi_{init}, A(d_1), \dots, A(d_{r-1})} \\
    &\hspace{0.7cm}= p_l\left(\exc{l}^{(r)}(\omega \bm{X}, \phi_{init}, A(d_1), \dots, A(d_{r-1}))\right)\cdot \mathbbm{1}\{u_l^{(r)} \neq n_l| A(d_1), \dots, A(d_{r-1})\}.
\end{align*}
We observe that the only dependence on $\omega$ in the above equation occurs only through $\exc{l}^{(r)}$. Due to Lemma~\ref{ScaleInvarianceCascade}, we have
\[\bm{\psi}^{(r)}(\omega \bm{X}, \phi_{init}, A(d_1),\dots, A(d_{r-1})) = \bm{\psi}^{(r)}(\bm{X}, \phi_{init}, A(d_1),\dots, A(d_{r-1}))\]
for any $\omega>0$. Hence, Equation~\eqref{probLFails} holds. Moreover, the first edge to fail is chosen uniformly at random, which also does not depend on $\omega$. Hence, the result follows by interatively conditioning on the set of next edges to fail and applying Equation~\eqref{probLFails}.

\end{proof}

In light of Lemmas~\ref{ScaleInvariancePl},~\ref{ScaleInvarianceCascade}, and~\ref{cascadeProbInv}, to understand all possible cascade behaviors on a given graph, it is sufficient to consider vertex weights, scaled by the largest vertex weight, i.e. $$(X_1,\dots,X_{\nv})/\max\{X_1,\dots,X_{\nv}\}\in [0,1]^{\nv}.$$
Having established this, we next illustrate another important property of the model, i.e., continuity of the flow capacity vector and the flow matrix with respect to vertex weights. 
\subsection{Continuity properties} \label{continuity}
The results derived in this section constitute an important ingredient in the proof of the main theorem. Specifically, they allow us to reduce the complexity of the problem by focusing solely on certain limiting scenarios of the vertex weight distribution, which is possible by applying the continuity argument. Motivated by the catastrophe principle, the limiting scenarios of interest are the cases when one vertex has a large weight and the weight of all other vertices is marginal.

In the first two lemmas, we show the continuity properties of the flow capacity vector and the flow matrix before the cascade and at an arbitrary stage of the cascade.

\begin{lemma}[Continuity of $\flowLimVec^{(0)}$ and $\flowM^{(0)}$]\label{continuityPlanningOp}
Consider a convergent sequence of vertex weight vectors $(\bm{X}^k)_{k\in \N}$ with limit $\bm{X}^*\neq 0$ such that $\bm{X}^k\geq \bm{X}^*$ for all $k$. Let $\flowLimVec^{(0)}(\bm{X})$ and $\flowM^{(0)}(\bm{X})$ denote the flow capacity vector and the flow matrix at stage 0, given the vertex weight vector $\bm{X}$. Then, \begin{equation}\lim_{k\rightarrow \infty}\flowLimVec^{(0)}(\bm{X}^k) = \flowLimVec^{(0)}(\bm{X}^*)\label{flowLimCont}\end{equation}
and 
\begin{equation}\lim_{k\rightarrow \infty}\flowM^{(0)}(\bm{X}^k) = \flowM^{(0)}(\bm{X}^*)\label{flowCont}.\end{equation}
    \end{lemma}
\begin{proof}
    We begin by considering the flow capacity vector. First, recall that $\bm{X}$ uniquely defines the net travel matrix $\bm{U}$. Moreover, $$\flowLim{(v,w)}^{(0)}(\bm{X}) = \max\{\varepsilon_{\min}\sum_{i = 1}^{\nv}X_i,\tau g_{(v,w)}(\bm{X}), \tau g_{(w,v)}(\bm{X})\}, \quad \SPM(\bm{X}) = \flowM^*(\bm{U}, \infty)$$ for $(v,w)\in \EE$ (Equation~\eqref{capacityPl}). As $\SPM(\bm{X})$ is the solution of a linear optimization problem, $\SPM(\bm{X})$ is a continuous function of the net travel matrix $\bm{U}$ \cite[Thm. 2]{Faisca_Dua_Pistikopoulos_2007}. 
    Moreover, since $\SPTot(\bm{X}) = \SPM(\bm{X})\bm{e}$ and $\bm{U}$ is a linear function of $\bm{X}$, $\SPTot(\bm{X})$ is a continuous function of $\bm{X}$. Thus, it follows that 
    \[\lim_{k\rightarrow \infty}\SPTot(\bm{X}^k) = \SPTot(\bm{X}^*).\]
Hence, $\flowLimVec^{(0)}(\bm{X})$ is a maximum of three functions continuous in $\bm{X}$. Since the maximum operator preserves continuity, we conclude that 
    \[\lim_{k\rightarrow \infty}\flowLimVec^{(0)}(\bm{X}^k) = \flowLimVec^{(0)}(\bm{X}^*).\]
Next, we show~\eqref{flowCont}, using the following steps:\begin{enumerate}
    \item[A)] Show, using compactness, that $\flowM^{(0)}(\bm{X}^k)$ is a convergent subsequence $\flowM^{(0)}(\bm{X}^{k_j})$.
    \item[B)] Construct a sequence $\widetilde{\flowM}^k$ with limit $\flowM^{(0)}(\bm{X}^*)$. 
    \item[C)] Show that $\flowM^{(0)}(\bm{X}^{k_j})$ converges to $\flowM^{(0)}(\bm{X}^*)$ by comparing $c_f(\flowM^{(0)}(\bm{X}^{k_j}))$ and $c_f(\widetilde{\flowM}^{k_j})$, where we use the optimality of $\flowM^{(0)}(\bm{X}^{k_j})$, feasibility of $\widetilde{\flowM}^{k_j}$ and the fact that $\widetilde{\flowM}^{k_j}\rightarrow \flowM^{(0)}(\bm{X}^*)$ as $j\rightarrow\infty$.
    \item[D)] Show that $\flowM^{(0)}(\bm{X}^k)\rightarrow \flowM^{(0)}(\bm{X}^*)$ as $k\rightarrow \infty$ because all convergent subsequences of $\flowM^{(0)}(\bm{X}^k)$ converge to $\flowM^{(0)}(\bm{X}^*)$.
\end{enumerate} Recall that $\flowM^{(0)}(\bm{X}) = \flowM^*(\bm{U}(\bm{X}), \flowLimVec^{(0)}(\bm{U}(\bm{X})))$.
For the sake of exposition, denote $\flowM^k:=\flowM^{(0)}(\bm{X}^k)$, $\flowLimVec^k:=\flowLimVec^{(0)}(\bm{X}^k)$, $\flowLimVec^*:=\flowLimVec^{(0)}(\bm{X}^*)$, and $\bm{U}^k = \bm{U}(\bm{X}^k)$. 

To show A), note that for each $\bm{X}^k$, the total flow on any edge is bounded by the sum of all commodities in the network $\xi^k$, given by $\xi^k := \sum_{v\in\V} X_v$. Let $\xi = \sup_{k\in \N} \xi^k$ and note that $\xi$ is well-defined because $(\bm{X}^k)_{k\in\N}$ is, by assumption, a convergent sequence and $\bm{U}$ is a linear function of $\bm{X}$. Thus, $\flowM^k\in [0,\xi]^{\nee\times \nv}$, which is compact. Hence, $\flowM^k$ has at least one convergent subsequence. We choose an arbitrary convergence subsequence $\flowM^{k_j}$ with limit $\flowM^{\infty}$ as $j\rightarrow \infty$. 

Next, we show part B), where we construct a feasible sequence of flows $(\widetilde{\flowM}^k)_{k\in \N}$ converging to $\flowM^{(0)}(\bm{X}^*)$. In particular, let \begin{equation}\widetilde{\flowM}^k: = \flowM^{(0)}(\bm{X}^*) + \SPM(\bm{X}^k - \bm{X}^*).\label{eq:alternativeSeq}\end{equation}
Note that $\SPM(\bm{X}^k - \bm{X}^*)$ exists because $\bm{X}^k - \bm{X}^*\geq 0$. Hence, applying Corollary~\ref{additivityCor} with $\bm{X}_1 = \bm{X}^*$ and $\bm{X}_2 = \bm{X}^k - \bm{X}^*$, we obtain that
$$\widetilde{\flowM}^k = \flowM^{(0)}(\bm{X}^*) + \SPM(\bm{X}^k) - \SPM(\bm{X}^*).$$  We observe that $\widetilde{\flowM}^k\geq 0$ because $\flowM^{(0)}(\bm{X}^*)\geq 0$ and $\SPM(\bm{X}^k - \bm{X}^*)\geq 0$. Moreover,
\[\incM\widetilde{\flowM}^k =\incM\flowM^{(0)}(\bm{X}^*) + \incM\SPM(\bm{X}^k) - \incM\SPM(\bm{X}^*)= \bm{U}^* + \bm{U}^k - \bm{U}^* = \bm{U}^k,\]
showing that sequence $\widetilde{\flowM}^k$ is indeed feasible in \eqref{optFlowProblem}. 
Lastly, we observe that $\widetilde{\flowM}^k\rightarrow \flowM^{(0)}(\bm{X}^*)$ as $k\rightarrow\infty$ where we again use the continuity of $\SPM$~\cite[Thm. 2]{Faisca_Dua_Pistikopoulos_2007}.  

We show part C) by comparing the costs of $\flowM^{k_j}$ and $\widetilde{\flowM}^{k_j}$ for all $j\in \N$, including the limit $j\rightarrow \infty$, where we exploit the fact that optimal solutions to $\eqref{optFlowProblem}$ attain the lowest cost. Due to optimality of $\flowM^{k_j}$, we know that for each $j\in \N$,
\begin{equation}c_f(\flowM^{k_j}, \flowLimVec^{k_j})\leq c_f(\widetilde{\flowM}^{k_j}, \flowLimVec^{k_j}).\label{costComp}\end{equation}
The next step is to show that $c_f(\widetilde{\flowM}^{k_j}, \flowLimVec^{k_j})\rightarrow c_f(\flowM^{(0)}(\bm{X}^*),\flowLimVec^*)$ as $j\rightarrow \infty$. Recall that $$c_f(\widetilde{\flowM}^{k_j},\flowLimVec^{k_j}) = \sum_{e\in \EE} d_e \widetilde{\flow[e]}^{k_j} + \frac{b_e}{\beta + 1} \flowLim{e}^{k_j} \left(\widetilde{\flow[e]}^{k_j}/\flowLim{e}^{k_j}\right)^{\beta+1}.$$ 
By construction, we know that $\widetilde{\flowTot}^{k_j}$ converges as $j\rightarrow \infty$ and due to $\eqref{flowLimCont}$, $\flowLimVec^{k_j}$ converges as $j\rightarrow \infty$. Hence, $\lim_{j\rightarrow\infty} c_f(\widetilde{\flowM}^{k_j}, \flowLimVec^{k_j})$ exist if $\lim_{j\rightarrow \infty}\widetilde{\flow[e]}^{k_j}/\flowLim{e}^{k_j}$ exists for all $e\in \EE$. This is indeed true because $\flowLim{e}^{k_j}\geq \varepsilon_{\min} \sum_{v\in \V}X_v^*>0$, which means that $\lim_{j\rightarrow \infty}\flowLim{e}^{k_j}\neq 0$, which, by the quotient rule for limits, implies that $\lim_{j\rightarrow\infty} \widetilde{f}_e^{k_j}/\flowLim{e}^{k_j} = \widetilde{f}_e^*/\flowLim{e}^*.$ 
Hence, $c_f(\tilde{\flowM}^{k_j}, \flowLimVec^{k_j})\rightarrow c_f(\flowM^{(0)}(\bm{X}^*), \flowLimVec^*)$ as $j\rightarrow \infty$. Next, we take limits from both sides in~\eqref{costComp} and obtain that
\begin{equation}\label{costInequality}c_f(\flowM^{\infty}, \flowLimVec^*)\leq c_f(\flowM^{(0)}(\bm{X}^*), \flowLimVec^*).\end{equation}
Moreover, for every $j$, $\flowM^{k_j}$ satisfies $\incM\flowM^{k_j} = \bm{U}^{k_j}$. Therefore, taking the limit from both sides, we obtain that $\incM\flowM^\infty = \bm{U}^*$. In other words, $\flowM^\infty$ is a feasible flow matrix for $\bm{X}^*$. Hence, by the uniqueness of the optimal solution and~\eqref{costInequality}, we obtain that $\flowM^\infty = \flowM^{(0)}(\bm{X}^*)$. Moreover, since the convergent sequence $\flowM^{k_j}$ was chosen arbitrarily, this implies that all convergent subsequences of $\flowM^k$ converge to the same limit $\flowM^{(0)}(\bm{X}^*)$. 

Finally, we show part D) by contradiction. Suppose that $\flowM^{k}$ does not converge to $\flowM^{(0)}(\bm{X}^*)$ as $k\rightarrow \infty$. This means that there exists some $\varepsilon>0$ such that for every $k^*\in \N$, there exists a $k\geq k^*$ with $||\flowM^k - \flowM^{(0)}(\bm{X}^*)||>\varepsilon$. Hence, there exists a subsequence $\flowM^{k_l}$, $l\in \N$, such that $||\flowM^{k_l} - \flowM^{(0)}(\bm{X}^*)||>\varepsilon$. Moreover, this subsequence lies in the compact space $[0,\xi]^{\nee\times\nv}$ and therefore has a convergent subsequence $\flowM^{k_{l_i}}$, $i\in \N$, with limit $\lim_{i\rightarrow \infty}\flowM^{k_{l_i}}\neq \flowM^{(0)}(\bm{X}^*)$. However, this is a contradiction as $\flowM^{k_{l_i}}$ is a subsequence of $\flowM^k$, and we have shown that all its convergent subsequences converge to $\flowM^{(0)}(\bm{X}^*)$. Hence, we conclude that Equation~\eqref{flowCont} holds.
\end{proof}

Next, we extend the continuity results to an arbitrary cascade stage $r$.
\begin{lemma}[Continuity of $\flowLimVec^{(r)}$ and $\flowM^{(r)}$] \label{continuityFlowr}
Consider a convergent sequence of vertex weight vectors $(\bm{X}^k)_{k\in \N}$ with limit $\bm{X}^*\neq 0$, such that $\bm{X}^k\geq \bm{X}^*$ for all $k$. Let $\flowLimVec^{(r)}(\bm{X},\phi_{init}, D)$ and $\flowM^{(r)}(\bm{X},\phi_{init},D)$ denote the flow limit vector and the flow matrix corresponding to the vector of vertex weights $\bm{X}$ at the $r$-th stage of cascade $D$ with initial congestion factor $\phi_{init}$ for some $r\in \{1,\dots,|D|\}$. Then, \begin{equation}\lim_{k\rightarrow \infty}\flowLimVec^{(r)}(\bm{X}^k,\phi_{init},D) = \flowLimVec^{(r)}(\bm{X}^*,\phi_{init},D)\label{flowLimContr}\end{equation}
and 
\begin{equation}\lim_{k\rightarrow \infty}\flowM^{(r)}(\bm{X}^k,\phi_{init},D) = \flowM^{(r)}(\bm{X}^*,\phi_{init},D)\label{flowContr}.\end{equation}
    \end{lemma}
\begin{proof}

We show that Equations~\eqref{flowLimContr} and~\eqref{flowContr} hold, using induction. First, we consider the base case of $r = 1$. For this choice of $r$, we show that Equation~\eqref{flowLimContr} holds. We note that due to Lemma~\eqref{continuityPlanningOp}, 
\[\lim_{k\rightarrow \infty} \flowLimVec^{(0)}(\bm{X}^k) =\flowLimVec^{(0)}(\bm{X}^*).\]
Furthermore, it follows from Equation~\eqref{nextFlowLim} that for any net travel matrix $\bm{X}$ and $e\in \EE$,
\[\flowLim{e}^{(1)}(\bm{X},\phi_{init},D)  = \begin{cases} \phi_{init}\flowLim{e}^{(0)}(\bm{X}) &\text{if } e \in D_1,\\
\flowLim{e}^{(0)}(\bm{X})  &\text{otherwise. }
\end{cases}\]
Hence, we obtain that for $e\in D_1$\[\lim_{k\rightarrow \infty}\flowLim{e}^{(1)}(\bm{X}^k,\phi_{init},D) = \lim_{k\rightarrow \infty}\phi_{init}\flowLim{e}^{(0)}(\bm{X}^k) = \phi_{init}\flowLim{e}^{(0)}(\bm{X}^*) =\flowLim{e}^{(1)}(\bm{X}^*,\phi_{init},D) \]
and for $e\in \EE/D_1$
\[\lim_{k\rightarrow \infty}\flowLim{e}^{(1)}(\bm{X}^k,\phi_{init},D) = \lim_{k\rightarrow \infty}\flowLim{e}^{(0)}(\bm{X}^k) = \flowLim{e}^{(0)}(\bm{X}^*) =\flowLim{e}^{(1)}(\bm{X}^*, \phi_{init},D) .\]
Thus, we conclude that Equation~\eqref{flowLimContr} holds for $r = 1$.

To show that Equation~\eqref{flowContr} holds, we use an argument similar to the one given in the proof of Lemma~\ref{continuityPlanningOp}, where we replace the sequence in Equation~\eqref{eq:alternativeSeq} with the following sequence of feasible flow matrices: \[\widetilde{\flowM}^k = \flowM^{(1)}(\bm{X}^*, {e_1},\phi_{init}) + \SPM(\bm{X}^k - \bm{X}^*, {e_1},\phi_{init}).\]
Because of the similarity of these approaches, the details of the proof are omitted. 

Now, consider some $r\in \{2,\dots,|D|\}$ and suppose that Equation~\eqref{flowLimContr} holds for $r- 1$. From Equation~\eqref{nextFlowLim}, it follows that for any vertex weight vector $\bm{X}$ and $e\in \EE$,
\[\flowLim{e}^{(r)}(\bm{X},\phi_{init},D)  = \begin{cases} \phi_e\left(\psi_e^{(r-1)}(\bm{X},\phi_{init},D)\right)\flowLim{e}^{(r-1)}(\bm{X},\phi_{init},D) &\text{if } e \in D_{r-1},\\
\flowLim{e}^{(r-1)}(\bm{X},\phi_{init},D)  &\text{otherwise. }
\end{cases}\]
Hence, for $e\in \EE/{D_{r-1}}$, we obtain
\[\lim_{k\rightarrow \infty}\flowLim{e}^{(r)}(\bm{X}^k,\phi_{init},D) = \lim_{k\rightarrow \infty}\flowLim{e}^{(r-1)}(\bm{X}^k,\phi_{init},D) = \flowLim{e}^{(r-1)}(\bm{X}^*,\phi_{init},D) = \flowLim{e}^{(r)}(\bm{X}^*,\phi_{init},D).\]
It remains to consider the case when $e\in D_{r-1}$. Recall that $\psi_e^{(r-1)} = \flow[e]^{(r-1)}/\flowLim{e}^{(r-1)}$. Moreover, our modeling assumptions, together with the fact that $\bm{X}^*\neq 0$, imply that $\lim_{k\rightarrow\infty } \flowLim{e}^{(r-1)}(\bm{X}^k,\phi_{init},D) > 0$. Hence, $\lim_{k\rightarrow\infty} \psi_e^{(r-1)}(\bm{X}^k,\phi_{init},D)$ exists and it is equal to
\[\lim_{k\rightarrow\infty} \psi_e^{(r-1)}(\bm{X}^k,\phi_{init},D) = \frac{\lim_{k\rightarrow \infty}\flow[e]^{(r-1)}(\bm{X}^k,\phi_{init},D)}{\lim_{k\rightarrow \infty}\flowLim{e}^{(r-1)}(\bm{X}^k,\phi_{init},D)} =\frac{\flow[e]^{(r-1)}(\bm{X}^*,\phi_{init},D)}{\flowLim{e}^{(r-1)}(\bm{X}^*,\phi_{init},D)} = \psi_e^{(r-1)}(\bm{X}^*,\phi_{init},D).\]
Next, we use the fact that $\phi_e$ is a continuous function to conclude that
\begin{align*}&\lim_{k\rightarrow\infty}\flowLim{e}^{(r)}(\bm{X}^k,\phi_{init},D) \\&\hspace{0.7cm}= \phi_e\left(\lim_{k\rightarrow \infty} \psi_{e}^{(r-1)}(\bm{X}^k,\phi_{init},D)\right)\lim_{k\rightarrow\infty}\flowLim{e}^{(r-1)}(\bm{X}^k,\phi_{init},D) = \flowLim{e}^{(r)}(\bm{X}^*,\phi_{init},D).\end{align*}
Thus, we conclude that Equation~\eqref{flowLimContr} holds for any $r$.

Again, to show that Equation~\eqref{flowContr} holds as well for $r = 1$, we use an analog proof to the one of Lemma~\ref{continuityPlanningOp}, where we replace the sequence in Equation~\eqref{eq:alternativeSeq} with the following sequence of feasible flow matrices: \[\widetilde{\flowM}^k = \flowM^{(r)}(\bm{X}^*,\phi_{init},D) + \SPM(\bm{X}^k - \bm{X}^*,\phi_{init},D).\]
Again, because of the similarity of these approaches, the details of the proof are omitted. 
\end{proof}
Finally, we provide a result on the continuity of cascade probability with respect to vertex weights. 
\begin{lemma}\label{cascadeCont}
    Consider a sequence of vertex weights $\bm{X}(\varepsilon) = (X_1(\varepsilon), \dots, X_{\nv}(\varepsilon))$ such that $\lim_{\varepsilon\downarrow 0}\bm{X}(\varepsilon) = \bm{X}^*\neq 0$ and $\bm{X}(\varepsilon)\geq \bm{X}^*$ for all $\varepsilon\geq 0$. Moreover, consider a particular cascade sequence $D = d$. Then, $\PR{D = d|\bm{X}(\varepsilon), \phi_{init}}$ is a right-continuous function in the neighborhood of $\varepsilon = 0$. Specifically, for all $\varepsilon>0$ sufficiently small, there exists a function $h(\varepsilon)$ with $\lim_{\varepsilon\downarrow 0} h(\varepsilon) = 0$ such that 
    \[\PR{D = d|\bm{X}(\varepsilon), \phi_{init}} \leq  \PR{D = d|\bm{X}^*, \phi_{init}} + h(\varepsilon).\]
\end{lemma}
\begin{proof}
 Let $A(d_r)$ be the event that edges in set $d_r$ experience disruption in the cascade stage $r$. From the proof of Lemma~\ref{cascadeProbInv}, we obtain the following expression for the cascade probability:
    \begin{align*}\begin{split}
        &\PR{D = d | \bm{X}(\varepsilon), \phi_{init}} \\
        &\hspace{0.4cm}= \PR{A(d_1) }\cdot \prod_{r = 2}^{|d|}\Big(\prod_{l\in d_r} p_l(\psi_l^{(r)}(A(d_1),\dots, A(d_{r-1}),\bm{X}(\varepsilon),  \phi_{init}))\cdot\mathbbm{1}\{u_l^{(r)} \neq n_l |A(d_1),\dots, A(d_{r-1})\} \\
        &\hspace{0.8cm}\prod_{l\not\in d_r}\max\left\{\left(1 - p_l(\psi_l^{(r)}(A(d_1),\dots, A(d_{r-1}),\bm{X}(\varepsilon),  \phi_{init}))\right), \mathbbm{1}\{u_l^{(r)} = n_l |A(d_1),\dots, A(d_{r-1})\} \right\}\Big) \\
        &\hspace{0.8cm}\prod_{l\in \EE}\max\left\{\left(1 - p_l(\psi_l^{(|d|+1)}(A(d_1),\dots, A(d_{|d|}),\bm{X}(\varepsilon),  \phi_{init}))\right), \mathbbm{1}\{u_l^{(|d|+1)} = n_l |A(d_1),\dots, A(d_{|d|})\} \right\}.
        \end{split}
    \end{align*}

Recall that by assumption, $p_e$ is a continuous function for all $e\in \EE$. Moreover, using Lemma~\ref{continuityFlowr}, we conclude that for every $r\in \{2,\dots,|d|+1\}$, $$\lim_{\varepsilon\downarrow 0}\psi^{(r)}(A(d_1), \dots, A(d_{r-1}),\bm{X}(\varepsilon),\phi_{init}) = \psi^{(r)}(A(d_1), \dots, A(d_{r-1}), \bm{X}^*,\phi_{init}).$$
Hence, by continuity of the product and the maximum operators, we conclude that 
\begin{align*}&\lim_{\varepsilon \downarrow0} \PR{D = d| \bm{X}(\varepsilon),\phi_{init}} = \PR{D = d | \bm{X}^*,\phi_{init}}.\end{align*}
In other words, for every $\varepsilon > 0$ sufficiently small, there exists some $h(\varepsilon)$ with $\lim_{\varepsilon\downarrow 0} h(\varepsilon) = 0$, such that
    \[\PR{D = d|\bm{X}(\varepsilon), , \phi_{init}} \leq  \PR{D=d|\bm{X}^*, \phi_{init}} + h(\varepsilon).\]
\end{proof}

\subsection{Catastrophe principle} \label{catastrophePrinc}

In this section, we prove the catastrophe principle, stated in Proposition~\ref{sbj}. However, in this proof, we require an additional technical result, which shows that the travel cost for any normalized vertex weight vector is bounded at every stage of the cascade. This is shown in the following lemma. 
\begin{lemma}
    For every cascade stage $r\in \N\cup \{0\}$, $$c_f\left(\flowM^{(r)}(\bm{x}, \phi_{init}, D), \flowLimVec^{(r)}(\bm{x}, \phi_{init}, D)\right)\leq M(r) \text{ for all } \bm{x}\in[0,1]^{\nv}, D \in \mathcal{D}, \phi_{init}\in [l_m, l_M], $$
    with \[M(r):= \nv \sum_{e\in \EE}\left(d_e+\frac{b_e\left(\tau (l_m)^r\right)^{-\beta}}{\beta+1}\right).\]
    Moreover, $M(r)$ is increasing in $r$. \label{BoundFlow}
\end{lemma}
\begin{proof}
    Since our results do not depend on the choice of normalized vertex weight, the cascade, or the initial capacity decrease factor, in this proof, we do not explicitly state the dependence on these random variables. First, we consider $r = 0$. From Equation~\eqref{capacityPl}, we know that $\flowLimVec^{(0)}\geq \tau\SPTot$, where $\SPTot = \SPM\bm{e}$ is the total shortest path flow vector. Due to the optimality of $\flowM^{(0)}$, we have that
    $$c_f\left(\flowM^{(0)},\flowLimVec^{(0)}\right)\leq c_f\left(\SPM,\flowLimVec^{(0)}\right) .$$
    Rewriting the right-hand side, we find that
    \begin{align*}
      &c_f\left(\SPM,\flowLimVec^{(0)}\right)= \sum_{e\in \EE}\left( d_e g_e + \frac{b_e}{\beta + 1} \flowLim{e}^{(0)}\left(g_e/\flowLim{e}^{(0)}\right)^{\beta+1}\right)\\
      &\hspace{0.7cm}\leq \sum_{e\in \EE}\left( d_e g_e + \frac{b_e}{\beta + 1} \tau g_e\tau^{-\beta-1}\right)\\
      &\hspace{0.7cm}\overset{(\star)}\leq \sum_{e\in \EE} \left(d_e\nv + \frac{b_e\nv\tau^{-\beta}}{\beta+1}\right) \\
      &\hspace{0.7cm}= \nv \sum_{e\in \EE}\left(d_e+\frac{b_e\tau^{-\beta}}{\beta+1}\right)= M(0),
    \end{align*}
where in $(\star)$ we used that the flow on each edge can be bounded by the the total amount of commodities; this is at most $\nv$, as we only consider normalized vertex weights. Hence, we conclude that $$c_f\left(\flowM^{(0)},\flowLimVec^{(0)}\right)\leq M(0).$$

Now, consider $r\in\N$. Using Equation~\eqref{nextFlowLim}, we obtain that \[\flowLim{e}^{(r)} \geq  (l_m)^r\flowLim{e}^{(0)}\] because at every stage of the cascade, the capacity may be lowered at most by a factor of $l_m$. With this, we find that the current cost of planning flow can be bounded by
    \begin{align*}
      &c_f\left(\SPM,\flowLimVec^{(r)}\right)= \sum_{e\in \EE}\left( d_e g_e + \frac{b_e}{\beta + 1} \flowLim{e}^{(r)}\left(g_e/ \flowLim{e}^{(r)}\right)^{\beta+1}\right)\\
      &\hspace{0.7cm}\leq \sum_{e\in \EE} \left(d_eg_e + \frac{b_e}{\beta+1} \tau (l_m)^r g_e\left(\frac{g_e}{\tau (l_m)^rg_e}\right)^{\beta+1}\right) \\
      &\hspace{0.7cm}\leq \sum_{e\in \EE} \left(d_e\nv + \frac{b_e\nv\left(\tau (l_m)^r\right)^{-\beta}}{\beta+1}\right) \\
      &\hspace{0.7cm}=\nv \sum_{e\in \EE}\left(d_e+\frac{b_e\left(\tau (l_m)^r\right)^{-\beta}}{\beta+1}\right)= M(r).
    \end{align*}
    Hence, we conclude that $$c_f\left(\flowM^{(r)},\flowLimVec^{(r)}\right)\leq c_f\left(\SPM,\flowLimVec^{(r)}\right)\leq M(r).$$
    Lastly, for any $r\in \N$ we find that
    \[M(r) - M(r-1) = \nv\sum_{e\in \EE} \frac{b_e}{\beta+1}\tau^{-\beta}\left(l_m^{-r\beta} - l_m^{-(r-1)\beta}\right) = \nv\sum_{e\in \EE} \frac{b_e}{\beta+1}\tau^{-\beta}l_m^{-(r-1)\beta}\left(l_m^{-\beta} - 1\right).\]
    Clearly, since $l_m\in (0,1)$ and $\beta, b_e, \nv, \tau\geq 0$, we find that $M(r) - M(r-1)\geq 0 $. Hence, we conclude that function $M(r)$ is increasing.
\end{proof}

\begin{proof}[Proof of Proposition~\ref{sbj}]
First, we show that $\Delta c_f^{(r)}\leq M(r)  X_{(\nv)}$. Let $\bm{X} =(X_1,\dots,X_{\nv}) $ and let $\bm{x}:=\bm{X} /X_{(\nv)}$. Choosing $\omega = X_{(\nv)}$, Corollary~\ref{scaleInvDelta} gives us
\[c_f\left(\flowM^{(k)}(\bm{X}, \phi_{init}, D), \flowLimVec^{(k)}(\bm{X}, \phi_{init}, D)\right) = X_{(\nv)} c_f\left(\flowM^{(k)}(\bm{x}, \phi_{init}, D), \flowLimVec^{(k)}(\bm{x}, \phi_{init}, D)\right).\]
  Now, applying Lemma~\ref{BoundFlow}, we find that for any cascade $D$ and initial capacity decrease factor $\phi_{init}$, we have the following bound on the flow cost at stage $r$:
\[c_f\left(\flowM^{(r)}, \flowLimVec^{(r)}\right)\leq M(r)X_{(\nv)}.\]
Moreover, since the flow cost function $c_f$ is nonnegative, we obtain
\begin{align}\begin{split}&\Delta c_f^{(r)} = c_f\left(\flowM^{(k)}(\bm{X}, \phi_{init}, D), \flowLimVec^{(k)}(\bm{X}, \phi_{init}, D)\right) - c_f\left(\flowM^{(0)}(\bm{X}, \phi_{init}, D), \flowLimVec^{(0)}(\bm{X}, \phi_{init}, D)\right)\\
&\hspace{0.7cm}\leq M(r)X_{(\nv)},\end{split}\label{eq:MrXbound}\end{align}
where we bounded the first summand by $M(r)X_{(\nv)}$ and the second by 0. Using the derived bound on the congestion cost, we obtain 
\begin{align*}
    &\PR{\Delta c_f^{(r)} >y, \sum_{i = 1}^{\nv} X_i - X_{(\nv)}\geq \varepsilon X_{(\nv)}} \\&\hspace{0.7cm}\leq \PR{M(r)X_{(\nv)}>y, \sum_{i = 1}^{\nv} X_i - X_{(\nv)}\geq \varepsilon X_{(\nv)}}\\
    &\hspace{0.7cm}\leq \PR{X_{(\nv)}>y/(M(r)), \sum_{i = 1}^{\nv} X_i - X_{(\nv)}\geq \varepsilon y/(M(r))}\\
    &\hspace{0.7cm}\leq \PR{X_{(\nv)}>y/(M(r)), (\nv-1)X_{(\nv-1)}\geq \varepsilon y/(M(r))},\\
\end{align*}
where $X_{(\nv- k)}$ is the $(k+1)$-th largest order statistic of $(X_1,\dots, X_{\nv})$. 
Furthermore, we know that for $k\in \{1,\dots,\nv\}$,
\begin{equation}
\PR{X_{(\nv)}>c_{\nv}y, \dots , X_{(\nv-k)}>c_{\nv-k} y} = \OO{y^{-k\alpha}}, \quad c_{\nv - k},\dots, c_{\nv}>0.    
\label{eq:orderStatAsymptotic}\end{equation}
For $k = 1$, Equation~\eqref{eq:orderStatAsymptotic} can be shown as follows:
\begin{align*}
    &\PR{X_{(\nv)}>c_{\nv} y , X_{(\nv - 1)} >c_{\nv - 1}y} \\&\hspace{0.7cm}\overset{(\star)}= n(n-1)\PR{X_{\nv}>c_{\nv} y , X_{\nv - 1} >c_{\nv - 1}y, X_1, \dots, X_{n-2}\leq X_{n-1}\leq X_{n}}\\
    &\hspace{0.7cm} \leq n(n-1)\PR{X_n> c_{\nv}y, X_{\nv-1}> c_{\nv - 1}y}\\
    &\hspace{0.7cm}= n(n-1)\PR{X_n> c_{\nv}y}\PR{X_{\nv-1}> c_{\nv - 1}y} = \OO{y^{-2\alpha}},
\end{align*}
where in $(\star)$ we use the fact that there are $n(n-1)$ possibilities for the largest two order statistics, each occurring with the same probability.
Hence, we conclude that 
\[ \PR{\Delta c_f^{(r)} >y, \sum_{i = 1}^{\nv} X_i - X_{(\nv)}\geq \varepsilon X_{(\nv)}} =  \OO{y^{-2\alpha}}.\]
\end{proof}
\begin{corollary} \label{sbj1}
    Consider a vector of vertex weights $(X_1,\dots, X_{\nv})$. Let $X_{(\nv)} := \max\{X_1,\dots, X_{\nv}\}$. Then, for all $\varepsilon>0$,
    \[\PR{ \Delta c_f^{(end)}>y , \sum_{i = 1}^n X_i - X_{(\nv)} \geq \varepsilon X_{(\nv)}} = \OO{y^{-2\alpha}}.
    \]
\end{corollary}
\begin{proof}
Every edge $e\in \EE$ can be disrupted at most $n_e$ times. This implies that the total number of cascade stages is bounded by $\sum_{e\in\EE}n_e$. Hence, 
\[\Delta c_f^{(end)} = \Delta c_f^{(\sum_{e\in\EE}n_e)}.\]
Thus, the result follows directly from Proposition~\ref{sbj}.
\end{proof}

\subsection{Proof of Theorem~\ref{mainTheorem}} \label{mainTheoremProof}
\begin{proof}
    We begin by proving Equation~\eqref{eq:mainThm}. For $r\in \N$, we observe that for all $\varepsilon\in (0,1)$, Lemma~\ref{sbj} yields
    \[\PR{\Delta c_f^{(r)}>y} = \PR{\Delta c_f^{(r)}>y, \sum_{i = 1}^{\nv- 1} X_{(i)}\leq \varepsilon X_{(\nv)}} +\OO{y^{-2\alpha}}.\] Hence, it suffices to consider the first summand. The additional travel cost $\Delta c_f^{(r)}$ depends on three aspects: 1) vertex weights, 2)  the capacity decrease caused by the first disruption, 3) the cascade. Given 1), 2), and 3), the additional travel cost at stage $r$ of the cascade is a deterministic function. Hence, in what follows, we wish to apply the law of total probability and add the probability of all possible events of 1), 2), and 3) to occur. The capacity decrease factor of the first edge disruption is a continuous random variable taking values in $[l_{m},l_{M}]\subset (0,1) $ with probability density function $p_{init}(x)$. Hence, summing over all possible cascades $d\in \mathcal{D}$ and conditioning on the value of $\phi_{init}$, we obtain
    \begin{align*}
&\PR{\Delta c_f^{(r)}>y, \sum_{i = 1}^{\nv - 1} X_{(i)}\leq \varepsilon X_{(\nv)}} \\&\hspace{0.7cm}= \sum_{d\in \mathcal{D}} \int_{l_{m}}^{l_{M}}\PR{\Delta c_f^{(r)}>y, \sum_{i = 1}^{\nv - 1} X_{(i)}\leq \varepsilon X_{(\nv)}, D = d~|~ \phi_{init} = u} p_{init}(u) \, \mathrm{d}u. 
    \end{align*}
Next, we use the fact that $\{X_i = X_{(\nv)}\} = \{X_i\geq X_j \text{ for all } j\in \V\}$ to obtain
\begin{align*}
    &\PR{\Delta c_f^{(r)}>y, \sum_{i = 1}^{\nv - 1} X_{(i)}\leq \varepsilon X_{(\nv)},D = d~|~ \phi_{init} = u}\\
    &\hspace{0.7cm} = \sum_{v\in \V} \PR{\Delta c_f^{(r)}>y, \sum_{w\neq v \in \V} X_w\leq \varepsilon X_v, X_v\geq X_w \text{ for all } w\in \V,D = d~|~ \phi_{init} = u}\\
    &\hspace{0.7cm}= \sum_{v\in \V} \PR{\Delta c_f^{(r)}>y, \sum_{w\neq v \in \V} X_w\leq \varepsilon X_v , D = d~|~ \phi_{init} = u},
\end{align*}
where in the last line we used the fact that $\varepsilon<1$. 

Next, we consider the above term for some fixed $v\in\V$. From Equation~\eqref{eq:MrXbound}, we know that $\Delta c_f^{(r)}$ is upper-bounded by $M(r)X_{v}$, where $X_{v}$ is the maximum of $X_1,\dots,X_{\nv}$. Hence, if $\varepsilon<1$ and $X_v<\frac{\varepsilon}{M(r)} y$, then $\Delta c_f^{(r)}\leq \varepsilon y < y$. Therefore, 
\begin{align*}&\PR{\Delta c_f^{(r)}>y, \sum_{w\neq v \in \V} X_w\leq \varepsilon X_v , D = d~|~ \phi_{init} = u} \\
&\hspace{0.7cm} = \PR{\Delta c_f^{(r)}>y, \sum_{w\neq v \in \V} X_w\leq \varepsilon X_v , D = d, X_v\geq \frac{\varepsilon}{M(r)}y~|~ \phi_{init} = u} \end{align*} To evaluate this probability further, let $h_{X_v}(x|w)$ denote the density function of $X_v$, conditional of the event $\{X_v>w\}$. With this, we obtain that 
\begin{align}\begin{split}\label{intEq2}&\PR{\Delta c_f^{(r)}>y, \sum_{w\neq v \in \V} X_w\leq \varepsilon X_v , D=d~\Big|~ \phi_{init} = u} \\
&\hspace{0.7cm}\overset{(A)}= \int_{\frac{\varepsilon}{M(r)}}^{\infty}\PR{\Delta c_f^{(r)}>y, \sum_{w\neq v \in \V} X_w\leq \varepsilon X_v , D=d~\Big|~ \phi_{init} = u,X_v = yz }  \\&\hspace{1.4cm}\cdot \PR{X_v\in \mathrm{d}(yz) ~|~\phi_{init} = u}\\
&\hspace{0.7cm}\overset{(B)}= \int_{\frac{\varepsilon}{M(r)}}^{\infty}\PR{\Delta c_f^{(r)}>y~\Big|~ \phi_{init} = u,X_v = yz , \sum_{w\neq v \in \V} X_w\leq \varepsilon X_v ,D=d} \\&\hspace{1.4cm}\cdot  \PR{\sum_{w\neq v \in \V} X_w\leq \varepsilon X_v ~\Big|~\phi_{init}, X_v = yz} \\
&\hspace{2.1cm}\cdot\PR{D=d~\Big|~ \phi_{init} = u,X_v = yz , \sum_{w\neq v \in \V} X_w\leq \varepsilon X_v}\cdot \PR{X_v\in \mathrm{d}(yx)}\\
&\hspace{0.7cm}\overset{(C)}= \PR{X_v>\frac{\varepsilon}{M(r)}y} \int_{\frac{\varepsilon}{M(r)}}^{\infty}\PR{\Delta c_f^{(r)}>\frac{1}{z}~\Big|~ \phi_{init} = u,X_v = 1 , \sum_{w\neq v \in \V} X_w\leq \varepsilon,D=d }  \\&\hspace{1.4cm}\cdot  \PR{\sum_{w\neq v \in \V} X_w\leq \varepsilon X_v ~\Big|~ X_v = yz} h_{X_v}\left(yz~\Big|~\frac{\varepsilon}{M(r)}y\right) \, \mathrm{d}z \,\\
&\hspace{2.1cm}\cdot\PR{D=d~\Big|~ \phi_{init} = u,X_v = 1, \sum_{w\neq v \in \V} X_w\leq \varepsilon}.\end{split}
\end{align}
Note that in step (A), we applied the law of total probability on $X_v$. In step (B), we applied the independence of $X_v$ and $\phi_{init}$, and we applied the conditional probability rule twice. Finally, in step (C), we conditioned on the event that $\{X_v>\frac{\varepsilon}{M(r)}y\}$ and applied the scale invariance property of the cost function (Corollary~\ref{scaleInvDelta}), independence of $\pareto$ and $\phi_{init}$, and the scaling property of the cascade probability (Lemma~\ref{cascadeProbInv}). 

Next, we evaluate the term $\PR{\Delta c_f^{(r)}>\frac{1}{z}~|~ \phi_{init} = u,X_v = 1 , \sum_{w\neq v \in \V} X_w\leq \varepsilon,D=d }$. Consider an arbitrary sequence $\bm{x}(\varepsilon)\geq \unitV{v}$ such that $\bm{x}(\varepsilon) \rightarrow \unitV{v}$ as $\varepsilon\downarrow 0$. As a consequence of Lemma~\ref{continuityFlowr} and the continuity of the flow cost function $\flowCost$, we obtain that
\begin{equation}\lim_{\varepsilon\downarrow 0}\Delta c_f^{(r)}(\bm{X} = \bm{x}(\varepsilon),\phi_{init} = u, D=d) =  \Delta c_f^{(r)}(\bm{X} = \unitV{v}, \phi_{init} = u,  D=d).\label{eq:30}\end{equation}
Moreover, from Lemma~\ref{BoundFlow}, we obtain that for all $\bm{x}(\varepsilon)\in [0,1],$
\begin{equation}\label{eq:31}\Delta c_f^{(r)}(\bm{X} = \bm{x}(\varepsilon),\phi_{init} = u, D=d) \leq M(r).\end{equation}
Hence, from \eqref{eq:30} and \eqref{eq:31}, it follows that for every $\varepsilon$ sufficiently small and $\bm{x}(\varepsilon)\in [0,1]$, there exists $M_{\bm{x}(\varepsilon)}(u,s)\leq 2M(r)$ with $\lim_{\varepsilon\downarrow0}M_{\bm{x}(\varepsilon)}(u,s) = 0$ such that
\[ \left| \Delta c_f^{(r)}(\bm{X} = \unitV{v}, \phi_{init} = u,  D = d) -  \Delta c_f^{(r)}(\bm{X} = \bm{x}(\varepsilon), \phi_{init} = u, D=d) \right|\leq  M_{\bm{x}(\varepsilon)}(u,s). \]
Now, for every $\varepsilon>0$, let $$M(\varepsilon, u,s) := \sup_{\bm{x}(\varepsilon)\in [0,1]^{\nv}, x_v(\varepsilon) = 1}M_{\bm{x}(\varepsilon)}(u,s),$$
which is well-defined, because we are taking a supremum over a set that is bounded from above by $2M(r)$. We use this to bound the probability that the cascade cost is larger than $1/z$. Let $Z^{(r)}(v,u,s)$ be the congestion cost at the $r-$th stage of the cascade, given that $X_v = 1$, $X_w = 0$ for $w\in \V/{v}$, $\phi_{init} = u$, and cascade $s$ has occurred, i.e., \begin{equation}Z^{(r)}(v,u,s) := \Delta c_f^{(r)}(\bm{X} = \bm{e}_v, \phi_{init} = u, D = d).\label{eq:DefZ}\end{equation} Using this definition, we find

\begin{align}
\begin{split}\label{bounds2}
   &\PR{Z^{(r)}(v,u, s) - M(\varepsilon, u, s)>\frac{1}{z}~\Big|~ \phi_{init} = u,X_v = 1 , \sum_{w\neq v \in \V} X_w\leq \varepsilon, D =d}\\&\hspace{0.7cm}\leq \PR{\Delta c_f^{(r)}>\frac{1}{z}~\Big|~ \phi_{init} = u,X_v = 1 , \sum_{w\neq v \in \V} X_w\leq \varepsilon, D = d }\\
    &\hspace{1.4cm}\leq \PR{Z^{(r)}(v,u, s) + M(\varepsilon, u, s)>\frac{1}{z}~\Big|~ \phi_{init} = u,X_v = 1 , \sum_{w\neq v \in \V} X_w\leq \varepsilon, D = d }.\end{split}\end{align}

Next, let $U(v,\varepsilon, u, s):=\left(Z^{(r)}(v,u,s) + M(\varepsilon, u,s)\right)^{-1}$ and $L(v,\varepsilon,u,s):=\left(Z^{(r)}(v,u, s) - M(\varepsilon, u,s)\right)^{-1}$. Observe that \[\lim_{\varepsilon\downarrow 0}U(v,\varepsilon, u, s) = \lim_{\varepsilon\downarrow 0}L(v,\varepsilon, u, s) = \left(Z^{(r)}(v,u,s)\right)^{-1}>0,\]
because every disruption cascade results in positive cost. Therefore, for all $\varepsilon$ small enough, we observe that $\varepsilon/(M(r))<U(v,\varepsilon,u,s)\leq L(v,\varepsilon,u,s)$. Hence, applying the upper bound in \eqref{bounds2} to \eqref{intEq2}, we find

\begin{align}\label{eq:32}\begin{split}
   & \PR{\Delta c_f^{(r)}>y, \sum_{w\neq v \in \V} X_w\leq \varepsilon X_v , D=d~|~ \phi_{init} = u} 
   \\
   &\hspace{0.7cm}\overset{\eqref{bounds2}}\leq \PR{X_v>\frac{\varepsilon}{M(r)}y}  \PR{D=d~|~ \phi_{init} = u,X_v = 1 , \sum_{w\neq v \in \V} X_w\leq \varepsilon}\\
   &\hspace{1.4cm}\cdot\int_{\frac{\varepsilon}{M(r)}}^{\infty}\PR{z>U(v,\varepsilon,u,s)~\Big|~ \phi_{init} = u,X_v = 1 , \sum_{w\neq v \in \V} X_w\leq \varepsilon,D=d } \\
   &\hspace{2.1cm}\cdot h_{X_v}\left(yz\Big|\frac{\varepsilon}{M(r)}y\right)\PR{\sum_{w\neq v \in \V} X_w\leq \varepsilon X_v ~|~ X_v = yz} \, \mathrm{d}z
   \\
   &\hspace{0.7cm} \overset{(\star)}= \PR{X_v>\frac{\varepsilon}{M(r)}y}  \PR{D=d~|~ \phi_{init} = u,X_v = 1 , \sum_{w\neq v \in \V} X_w\leq \varepsilon}\, \\
&\hspace{1.4cm}\cdot\int_{U(v,\varepsilon, u, s)}^{\infty}h_{X_v}\left(yz\Big|\frac{\varepsilon}{M(r)}y\right)\PR{\sum_{w\neq v \in \V} X_w\leq \varepsilon X_v ~|~ X_v = yz} \, \mathrm{d}z\\
   &\hspace{0.7cm} \overset{(\star\star)}\leq \PR{X_v>\frac{\varepsilon}{M(r)}y}  \PR{D=d~|~ \phi_{init} = u,X_v = 1 , \sum_{w\neq v \in \V} X_w\leq \varepsilon}\, \\
&\hspace{1.4cm}\cdot\int_{U(v,\varepsilon, u, s)}^{\infty}h_{X_v}\left(yz\Big|\frac{\varepsilon}{M(r)}y\right) \, \mathrm{d}z\\
   &\hspace{0.7cm} \overset{(\star\star\star)}= \PR{X_v> U(v,\varepsilon, u, s)y} \cdot\PR{D=d~|~ \phi_{init} = u,X_v = 1 , \sum_{w\neq v \in \V} X_w\leq \varepsilon}.
\end{split}
\end{align}
Note that in $(\star)$, we used the fact that $\PR{z>U(v,\varepsilon,u,s) ~\big|~ \phi_{init} = u, X_v = 1, \sum_{w\neq v \in \V} X_w\leq \varepsilon,D=d } = 1$ if $z>U(v,\varepsilon,u,s)$ and 0 otherwise, allowing us to change the integration boundary. In $(\star\star)$, we bounded the last term from above by 1, and in $(\star\star\star)$ we removed the conditioning and applied the law of total probability on $X_v$.

For the lower bound, we take a similar approach and obtain the following expression:
\begin{align}\label{eq:33}\begin{split}
   & \PR{\Delta c_f^{(r)}>y, \sum_{w\neq v \in \V} X_w\leq \varepsilon X_v , D = d~|~ \phi_{init} = u} \\
&\hspace{0.7cm}\overset{\eqref{bounds2}}\geq \PR{X_v>\frac{\varepsilon}{M(r)}y}  \PR{D=d~|~ \phi_{init} = u,X_v = 1 , \sum_{w\neq v \in \V} X_w\leq \varepsilon}\\
   &\hspace{1.4cm}\cdot\int_{\frac{\varepsilon}{M(r)}}^{\infty}\PR{z>L(v,\varepsilon,u,s)~\Big|~ \phi_{init} = u,X_v = 1 , \sum_{w\neq v \in \V} X_w\leq \varepsilon,D=d } \\
   &\hspace{2.1cm}\cdot h_{X_v}\left(yz\Big|\frac{\varepsilon}{M(r)}y\right)\PR{\sum_{w\neq v \in \V} X_w\leq \varepsilon X_v ~|~ X_v = yz} \, \mathrm{d}z
   \\
      &\hspace{0.7cm} = \PR{X_v>\frac{\varepsilon}{M(r)}y}\PR{D=d~|~ \phi_{init} = u,X_v = 1 , \sum_{w\neq v \in \V} X_w\leq \varepsilon} \,\\
&\hspace{1.4cm}\cdot\int_{L(v,\varepsilon, u, s)}^{\infty}h_{X_v}\left(yz\Big|\frac{\varepsilon}{M(r)}y\right)\PR{\sum_{w\neq v \in \V} X_w\leq \varepsilon X_v ~|~ X_v = yz} \, \mathrm{d}z \\
   &\hspace{0.7cm} \overset{(\star\star\star\star)}\geq \PR{X_v>\frac{\varepsilon}{M(r)}y} \PR{D=d~|~ \phi_{init} = u,X_v = 1 , \sum_{w\neq v \in \V} X_w\leq \varepsilon} \,\\
&\hspace{1.4cm}\cdot\int_{L(v,\varepsilon, u, s)}^{\infty}h_{X_v}\left(yz\Big|\frac{\varepsilon}{M(r)}y\right)\PR{\sum_{w\neq v \in \V} X_w\leq \varepsilon L(v,\varepsilon, u, s)y} \, \mathrm{d}z\\
   &\hspace{0.7cm}=\PR{X_v> L(v,\varepsilon,u, s)y}\PR{D = d~|~ \phi_{init} = u,X_v = 1 , \sum_{w\neq v \in \V} X_w\leq \varepsilon}\, \\
   &\hspace{1.4cm} \cdot \PR{\sum_{w\neq v \in \V} X_w\leq \varepsilon L(v,\varepsilon, u, s)y},
\end{split}
\end{align}
where in $(\star\star\star\star)$, we use the fact that $X_v\geq L(v,\varepsilon, u,s)y$ and then apply the independence of $X_w$ and $X_v$ for all $w\neq v\in \V$. 

We observe that $\PR{\sum_{w\neq v \in \V} X_w\leq \varepsilon L(v,\varepsilon, u, s)y}\rightarrow 1$ as $y\rightarrow \infty$ because $\varepsilon L(v,\varepsilon, u,s)$ is a constant w.r.t.\ $y$, so the right-hand side of the inequality approaches $\infty$. Hence, it follows from \eqref{eq:32} and \eqref{eq:33}, and the fact that $X_v$ is Pareto-tailed, that
\begin{align}\begin{split}
    &K\cdot L(v,\varepsilon, u, s)^{-\alpha}\cdot \PR{D = d~|~ \phi_{init} = u,X_v = 1 , \sum_{w\neq v \in \V} X_w\leq \varepsilon} \\
    &\hspace{0.7cm}\leq \lim_{y\rightarrow\infty} y^{\alpha } \cdot \PR{\Delta c_f^{(r)}>y, \sum_{w\neq v \in \V} X_w\leq \varepsilon X_v , D = d~|~ \phi_{init} = u} \\
    &\hspace{0.7cm}\leq K\cdot U(v,\varepsilon, u, s)^{-\alpha} \cdot \PR{D=d~|~ \phi_{init} = u,X_v = 1 , \sum_{w\neq v \in \V} X_w\leq \varepsilon}.\label{boundsEps1}\end{split}\end{align}

Both bounds in Equation~\eqref{boundsEps1} hold for all $\varepsilon$ sufficiently small. Moreover, recall that from definitions of $U(v,\varepsilon,u,s)$ and $L(v,\varepsilon,u,s)$, it follows that $U(v,\varepsilon,u,s)^{-1}\rightarrow Z^{(r)}(v,u,s)$ and $L(v,\varepsilon,u,s)^{-1}\rightarrow Z^{(r)}(v,u,s)$ as $\varepsilon\downarrow 0$. Hence, taking the limit $\varepsilon\downarrow0$ and applying Lemma~\ref{cascadeCont}, we obtain
\begin{align*}&\lim_{\varepsilon \downarrow 0}\lim_{y\rightarrow \infty} y^{\alpha } \cdot \PR{\Delta c_f^{(r)}>y, \sum_{w\neq v \in \V} X_w\leq \varepsilon X_v , D=d~|~ \phi_{init} = u} \\
&\hspace{0.7cm}= K\cdot Z^{(r)}(v,u, s)^{\alpha} \cdot \PR{D =d| \phi_{init} = u, \bm{X} = \unitV{v}}.\end{align*}
Thus, altogether we obtain that
\begin{align} \label{eq:proofCr}
    \begin{split}
        &\lim_{y\rightarrow \infty}y^{\alpha}\PR{\Delta c_f^{(r)}>y} \\
        &\hspace{0.7cm}= \lim_{y\rightarrow \infty}y^{\alpha}\sum_{v\in \V}\sum_{d\in \mathcal{D}}\int_{l_{m}}^{l_{M}}\PR{c_f^{(r)}>y, \sum_{w\neq v \in \V} X_w \leq \varepsilon X_v, D=d ~|~ \phi_{init} = u}p_{init}(u) \, \mathrm{d}u + \OO{y^{-\alpha}} \\
        &\hspace{0.7cm}=\sum_{v\in \V}\sum_{d\in \mathcal{D}}\int_{l_{m}}^{l_{M}} p_{init}(u)\lim_{\varepsilon\downarrow 0}\lim_{y\rightarrow \infty}y^{\alpha}\cdot \PR{c_f^{(r)}>y, \sum_{w\neq v \in \V} X_w\leq \varepsilon X_v, D=d|\phi_{init} = u} \, \mathrm{d}u\\
        &\hspace{0.7cm}= \sum_{v\in \V}\sum_{d\in \mathcal{D}}K\int_{l_{m}}^{l_{M}} Z^{(r)}(v,u, s)^{\alpha}\cdot\PR{D=d | \phi_{init} = u, \bm{X} = \unitV{v}}\cdot p_{init}(u)\, \mathrm{d}u
        .
    \end{split}
\end{align}
Recall that, conditioned on $\pareto, \phi_{init},$ and $D$, the congestion cost $\Delta c_f^{(r)}$ is no longer random, but deterministic. Hence, the following expression is true: \[\left(Z^{(r)}(v,u,s)\right)^\alpha\overset{\eqref{eq:DefZ}} = \left(\Delta c_f^{(r)}(\pareto = \unitV{v},\phi_{init} = u, D = d ) \right)^{\alpha}=  \E\left[\left(\Delta c_f^{(r)}\right)^\alpha \,\Big|\, \pareto = \unitV{v}, \phi_{init} = u, D = d\right].\] 
Therefore, applying the law of total expectation twice, Equation~\eqref{eq:proofCr} yields
\begin{align*} 
    \begin{split}
        &\lim_{y\rightarrow \infty}y^{\alpha}\PR{\Delta c_f^{(r)}>y} \\
        &\hspace{0.7cm}= K\sum_{v\in \V}\sum_{d\in \mathcal{D}}\int_{l_{m}}^{l_{M}} \E\left[\left(\Delta c_f^{(r)}\right)^\alpha\Big| \pareto = \unitV{v}, \phi_{init} = u, D = d\right]\PR{D=d | \phi_{init} = u, \bm{X} = \unitV{v}}\cdot p_{init}(u)\, \mathrm{d}u
        \\
        &\hspace{0.7cm}= K\sum_{v\in \V}\int_{l_{m}}^{l_{M}} \E\left[\left(\Delta c_f^{(r)}\right)^\alpha\Big| \pareto = \unitV{v}, \phi_{init} = u\right]\cdot p_{init}(u)\, \mathrm{d}u\\
        &\hspace{0.7cm}= K\sum_{v\in \V}\E\left[\left(\Delta c_f^{(r)}\right)^\alpha\Big| \pareto = \unitV{v}\right] = C^{(r)}.
    \end{split}
\end{align*}
Note that $C^{(r)}>0$, because $K>0$ and every disruption cascade necessarily increases the flow cost at every stage of the cascade, implying that $\Delta c_f^{(r)}>0$. Hence, Equation~\eqref{eq:mainThm} holds.

To prove Equation~\eqref{eq:DeltaEnd}, we use the fact that an arbitrary cascade has a finite number of stages bounded by $\sum_{e\in \EE} n_e$. Hence, for $r^* = \sum_{e\in \EE} n_e$, we have that $\Delta c_f^{(end)} = \Delta c_f^{(r^*)}$. This means that 
\[\PR{\Delta c_f^{(end)}>y} = \PR{\Delta c_f^{(r^*)}>y}\sim C^{(r^*)}y^{-\alpha}.\]
Moreover for all $r\geq r^*$, $v\in \V$, $d\in \mathcal{D}$, and $u\in [l_m, l_M]$ we have that $Z^{(r)}(v,u,s) = Z^{(r^*)}(v,u,s)$. Hence, 
\[\lim_{r\rightarrow \infty} Z^{(r)}(v,u,s) =Z^{(r^*)}(v,u,s)\]
and 
\[\lim_{r\rightarrow \infty} C^{(r)} = C^{(r^*)}>0.\]
This completes the proof of Theorem \ref{mainTheorem}.

\end{proof}
\subsection{Strict convexity of the flow cost function} \label{appStrictConvex}
In this section, we show that the flow cost function is strictly convex when $\flowLimVec\in \R^{\nee}$ and $\beta>0$. To this end, we first show a more general result, given in the following lemma.
\begin{lemma}
   A function $c$: $\R^{n}_+\rightarrow \R_+$ of the form
   \[c(\bm{x}) = \sum_{i = 1}^n a_i x_i + \frac{b_i}{\beta+1}x_i^{\beta+1},  \]
   for some constants $a_i, b_i, \beta>0$ for all $i = 1,\dots, n$, is strictly convex.\label{LemmaStrictConvex}
\end{lemma}
\begin{proof}
    Let $c_i(x_i):= a_ix_i + \frac{b_i}{\beta + 1} x_i^{\beta + 1}$ and observe that $c(\bm{x}) = \sum_{i = 1}^n c_i(x_i)$.  We show that $c(\bm{x})$ is strictly convex by first showing that $c_i(x_i)$ is strictly convex for all $i\in \{1,\dots, n\}$. From the definition of $c_i$, we obtain that
    \[\frac{d^2}{dx_i^2}c_i(x_i) = b_i\beta x_i^{\beta - 1}.\]
    Since $b_i, \beta>0$ and $x_i\in \R_+$, we find that $\frac{d^2}{dx_i^2}c_i(x_i)>0$ for all $x_i>0$. Hence, from the second-order criterion \cite{boyd2004convex}, we conclude that $c_i(x_i)$ is strictly convex for all $i\in \{1,\dots, n\}$. It remains to show that $c(\bm{x})$ is strictly convex, which we show from the definition. Let $\xi\in [0,1]$ and $\bm{x},\bm{y}\in \R^n_+$ where $\bm{x}\neq\bm{y}$. Then,
\begin{align*}
&c(\xi \bm{x} + (1-\xi)\bm{y}) = \sum_{i = 1}^n c_i(\xi x_i + (1-\xi)y_i)< \sum_{i = 1}^n \xi c_i(x_i) + (1-\xi)c_i(y_i) = \xi c(\bm{x}) + (1-\xi)c(\bm{y}),
\end{align*}
where we applied the definition of strict convexity for functions $c_i$. Hence, by definition, $c(\bm{x})$ is strictly convex as well. 

\end{proof}
Since for any fixed flow capacity vector $\flowLimVec$, the cost function $c_f$ defined in~\eqref{costEquation} satisfies the requirements of Lemma~\ref{LemmaStrictConvex}, we conclude that $c_f$ is strictly convex.  

\subsection{Optimal Flow Problem with linear cost function} \label{appendixSP}
In this section, we discuss the Optimal Flow Problem~\eqref{optFlowProblem} with linear cost function $c_f(\flowM) = \sum_{e\in \EE}d_ef_e$ in more detail. Recall that this linear problem arises as the Optimal Flow Problem when the capacity vector $\flowLimVec = \infty$, i.e.,  $\flowM^*(\bm{X}, \infty)$. Here, we discuss this problem from a different perspective, in particular, we show that this problem yields a flow distribution where any $v-w$ traffic follows only the shortest $v-w$ paths. 

Before we make this statement more precise, we define the notions of $v-w$ path, and the shortest $v-w$ path. 
\begin{definition}[$v - w$ path]
A $v-w$ path is a sequence of edges  $\{e_1 = (v,z_1), e_2 = (z_1,z_2) \dots,e_{l-1} = (z_{l-2}, z_{l-1}), e_l = (z_{l-1}, w)\}$ such that $e_i\in \EE$ for all $i\in \{1,\dots, l\}$. The set of all $v-w$ paths is denoted by $\mathcal{P}(v,w)$. 
\end{definition}
\begin{definition}[Shortest $v-w$ path]
A $v-w$ path $p(v,w)\in \mathcal{P}(v,w)$ is the shortest $v-w$ path if $$c(p(v,w)) = \min_{p\in \mathcal{P}(v,w)} c(p),$$
where $c(p):= \sum_{e\in p} d_e$. Furthermore, the set of all shortest $v-w$ paths is denoted by $\mathcal{SP}(v,w)$. 
\end{definition}
Note that since we assume that $d_e>0$ for all $e\in \EE$, there are finitely many shortest paths between any pair of vertices, i.e., $|\mathcal{SP}(v,w)|<\infty$ for all $v,w\in \V$. 

With these definitions, we can state the first result, which shows that any optimal solution to \eqref{optFlowProblem} with the linear cost function yields flow distribution that follows only shortest paths for every pair of vertices $v,w\in \V$. 
\begin{lemma}\label{SP}
Let $\flowM^*(\bm{X})$ be a solution to Problem~\eqref{optFlowProblem} with cost function $c_f(\flowM) = \sum_{e\in \EE}d_ef_e$ and consider the unique decomposition of $\flowM^*(\bm{X})$ into $v-w$ paths $p(v,w)$ with corresponding flow intensity $f(p)$ on path $p$. Then, for every $v,w\in \V$ and every $p\in \mathcal{P}(v,w),$
\[f(p)>0 \implies p\in \mathcal{SP}(v,w).\]
\end{lemma}
\begin{proof}
    We prove this by contradiction. Suppose that for some $v,w\in \V$, there exists $p\in \mathcal{P}(v,w)\setminus\mathcal{SP}(v,w)$ with $f(p)>0$. Next, choose some $\tilde{p}\in \mathcal{SP}(v,w)$ and construct a new flow $\tilde{\flowM}$ by rerouting the flow on path $p$ to path $\tilde{p}$. In particular, we set $\tilde{f}(p) = 0$, $\tilde{f}(\tilde{p}) = f(p) +  f(\tilde{p})$ and $\tilde{f}(r) = f(r)$ for every path $r\in \cup_{(v,w)\in V}\mathcal{P}(v,w)/\{p, \tilde{p}\}$. Note that $\tilde{\flowTot}$ yields a feasible flow distribution. 
    
    Next, since $p\not\in \mathcal{SP}(v,w)$, we know that
    \begin{equation}\label{costPtilde}\sum_{e\in p} d_e > \sum_{e\in \tilde{p}} d_e.\end{equation}
    Hence, by the linearity of the cost function, we find that 
    \[c_f(\tilde{\flowM}) = c_f(\flowM^*(\bm{X})) - \sum_{e\in p}d_e f(p) + \sum_{e\in \tilde{p}} d_e f(p) = c_f(\tilde{\flowM}) + f(p)\left(\sum_{e\in \tilde{p}} d_e - \sum_{e\in p} d_e\right)< c_f(\flowM^*(\bm{X})), \]
    where in the last step we apply \eqref{costPtilde} and the fact that $f(p)>0$. However, this is a contradiction because $\flowM^*(\bm{X})$ is optimal. 
\end{proof}
In view of the shortest path interpretation given in Lemma~\ref{SP}, it becomes apparent that Problem~\eqref{optFlowProblem} with input $(\bm{X}, \infty)$ may have multiple optimal solutions; this is why in our model we assume that $\flowM^*(\bm{X}, \infty)$ is the shortest path flow that assigns equal flow intensity to all shortest $v-w$ paths for all $v,w\in \V$. In particular, for all $v\neq w\in \V$ and all $p\in \mathcal{SP}(v,w)$, $f(p) = u_{v,w}(\bm{X})/|\mathcal{SP}(v,w)|$. Under this assumption, $\flowM^*(\bm{X}, \infty)$ is uniquely defined and we can show that it has an additive property, as stated in the following corollary.
\begin{corollary}\label{additivityCor}
    For any $\bm{X}_1,\bm{X}_2\in \R_+^{\nv}$, the Wardrop UE flow with unlimited capacity can be decomposed as follows
\begin{equation}\flowM^*(\bm{X}_1 + \bm{X}_2, \infty) = \flowM^*(\bm{X}_1, \infty) + \flowM^*(\bm{X}_2, \infty).\label{additivity1}\end{equation}
\end{corollary}
\begin{proof}
    Let $f(p, \bm{X})$ denote the flow intensity on path $p\in \cup_{(v,w)\in V}\mathcal{P}(v,w)$ corresponding to $\flowM^*(\bm{X}, \infty)$. By construction, all shortest $v-w$ paths admit the same flow intensity, which means that for any $p\in \mathcal{P}(v,w)$, \[f(p, \bm{X}) = u_{v,w}(\bm{X})/|\mathcal{SP}(v,w)|.\]
    Moreover, we observe that for all $e\in \EE, v\in \V$,\[f^*_{e,v}(\bm{X}, \infty) = \sum_{w\in \V}\sum_{p\in \mathcal{SP}(v,w): e\in p} f(p, \bm{X}).\]
    Hence, we find that
    \begin{align*}
        &f^*_{e,v}(\bm{X}_1 + \bm{X}_2, \infty) = \sum_{w\in \V}\sum_{p\in \mathcal{SP}(v,w): e\in p} f(p, \bm{X}_1 + \bm{X}_2)\\
        &\hspace{0.7cm}=\sum_{w\in \V}\sum_{p\in \mathcal{SP}(v,w): e\in p} u_{v,w}(\bm{X}_1 + \bm{X}_2)/|\mathcal{SP}(v,w)|\\
        &\hspace{0.7cm}=\sum_{w\in \V}\sum_{p\in \mathcal{SP}(v,w): e\in p} u_{v,w}(\bm{X}_1)/\mathcal{SP}(v,w) + u_{v,w}(\bm{X}_2)/|\mathcal{SP}(v,w)|\\
        &\hspace{0.7cm}= f^*_{e,v}(\bm{X}_1,\infty) + f^*_{e,v}(\bm{X}_2,\infty), 
    \end{align*}
where we use the fact that $\bm{U}(\bm{X})$ is a linear function of $\bm{X}$. Hence, Equation~\eqref{additivity1} holds. 
\end{proof} 
\end{document}